\documentclass[12pt]{article}

\usepackage[a4paper,margin=1in]{geometry}
\usepackage[authoryear,round]{natbib}
\usepackage{changepage}
\usepackage[colorlinks,citecolor=blue,linkcolor=blue,urlcolor=blue]{hyperref}

\usepackage{lipsum}

\usepackage{setspace}
\usepackage{pdfpages}
\usepackage[english]{babel}
\usepackage[en-GB]{datetime2}
\usepackage[T1]{fontenc}
\usepackage[utf8]{inputenc}
\usepackage{csquotes}

\newcommand{\OMIT}[1]{{\bf [OMIT:} #1 \ {\bf --- end OMIT] }}  
\renewcommand{\OMIT}[1]{}            

\usepackage{amsmath,amssymb,amsthm}
\usepackage{fixmath}
\usepackage{lmodern}
\usepackage{bm}
\allowdisplaybreaks

\usepackage{graphicx}
\usepackage{caption,subcaption}
\usepackage{float,subfloat}
\graphicspath{{pictures/}}

\usepackage{tikz}
\usepackage{pgfplots}
\usetikzlibrary{plotmarks}
\usetikzlibrary{shapes}
\usetikzlibrary{patterns}

\newtheorem{theorem}{Theorem}
\newtheorem{lemma}{Lemma}
\newtheorem{proposition}[theorem]{Proposition}

\newtheorem{conjecture}[theorem]{Conjecture}

\newtheorem{definition}{Definition}
\newtheorem{remark}{Remark}

\newenvironment{widerequation*}{%
    \begin{adjustwidth}{-1.5cm}{-1.5cm}\begin{equation*}}
    {\end{equation*}\end{adjustwidth}}

\newcounter{savefootnote}
\newcounter{symfootnote}
\newcommand{\symfootnote}[1]{%
   \setcounter{savefootnote}{\value{footnote}}%
   \setcounter{footnote}{\value{symfootnote}}%
   \ifnum\value{footnote}>30\setcounter{footnote}{0}\fi%
   \let\oldthefootnote=\thefootnote%
   \renewcommand{\thefootnote}{\fnsymbol{footnote}}%
   \footnote{#1}%
   \let\thefootnote=\oldthefootnote%
   \setcounter{symfootnote}{\value{footnote}}%
   \setcounter{footnote}{\value{savefootnote}}%
}

\begin{document}

\title{Parameter Sensitivity Analysis of Hierarchical Spatial Economy: 
Trade Strategy around Brexit%
\thanks{This work has received Grants-in-Aid from JSPS 21K04299/24K22968/24K16372/25K01339/25H00543 and Fusion Oriented Research for Disruptive Science and Technology (Grant No.\ JPMJFR215M).}}

\author{%
Kiyohiro Ikeda\thanks{Department of Civil and Environmental Engineering, Tohoku University, Sendai, Japan.}\quad ~
Yosuke Kogure\thanks{Civil and Environmental Engineering Course, Akita University, Akita, Japan.}\quad ~
Hiroki Aizawa\thanks{College of Transdisciplinary Sciences for Innovation, Kanazawa University, Ishikawa, Japan.}\quad ~
Yuki Takayama\thanks{Corresponding author. Department of Civil and Environmental Engineering, Institute of Science Tokyo, Tokyo, Japan. \texttt{takayama.y.cc65@m.isct.ac.jp}.}%
}

\date{\today}

\maketitle

\begin{abstract}
\noindent
This paper presents a systematic framework for analyzing the economic parameter sensitivity of a hierarchical spatial economy within economic geography models. Through the hierarchical reduction approach proposed in this study, the original region-level governing equation is condensed into country-level and alliance-level equations. Based on the reduced governing equation, we formulate the sensitivity of economic variables on each country’s population. This approach is applied to the analysis of international trade competition---covering both trade liberalization and protectionism around Brexit---among the UK, France, and Germany. We find that both the UK and the EU should focus on reducing domestic transportation costs, whereas tariffs and retaliatory tariffs act as a double-edged sword that can either strengthen or weaken their trade positions.



\end{abstract}




\medskip

\newpage


\section{Introduction}

Population distributions in hierarchical systems of countries and regions have been extensively studied in economic geography.
  The impacts of reductions in tariffs 
  and transportation costs between countries
  on within-country population distributions are analyzed extensively.\footnote{%
See, e.g., \citet{Krugman.Elizondo.1996,Crozet.Soubeyran.2004}, 
  \citet{Behrens.etal.2006,Behrens.etal.2007},
and \citet{Gallego.Zofio.2018}.
} 

 \begin{figure}
   \centering\small
   \begin{tabular}{cc}
\includegraphics[scale=0.39]{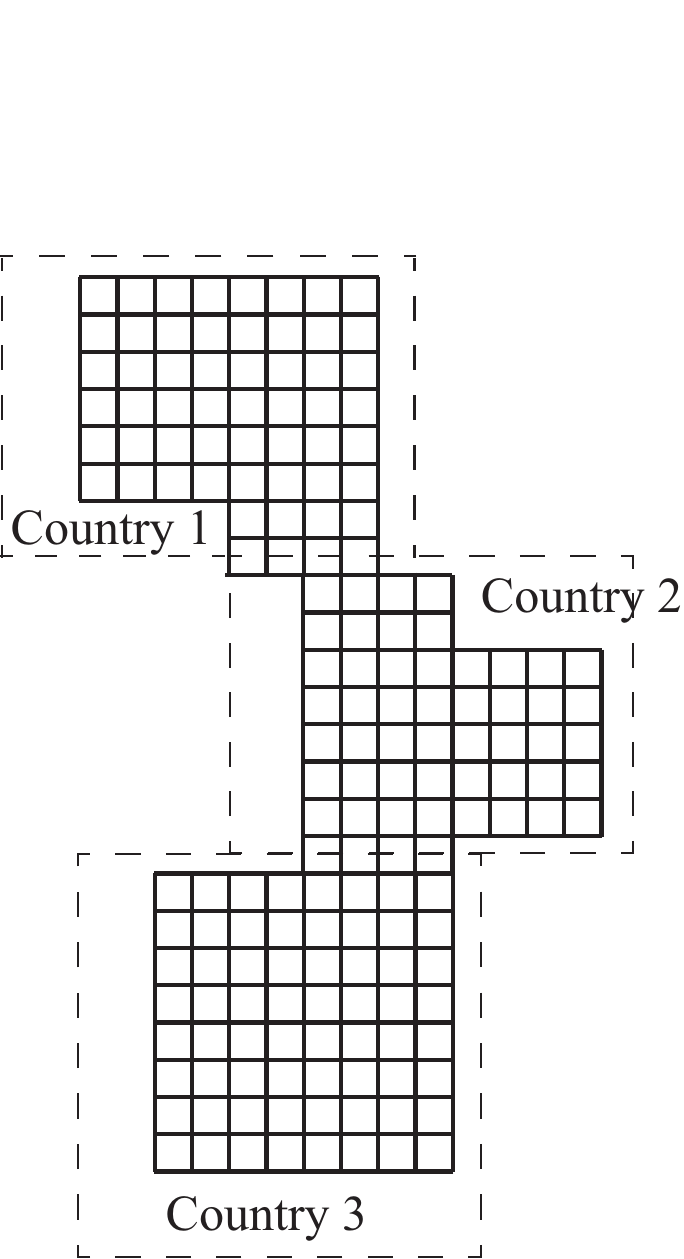} &
\includegraphics[scale=0.37]{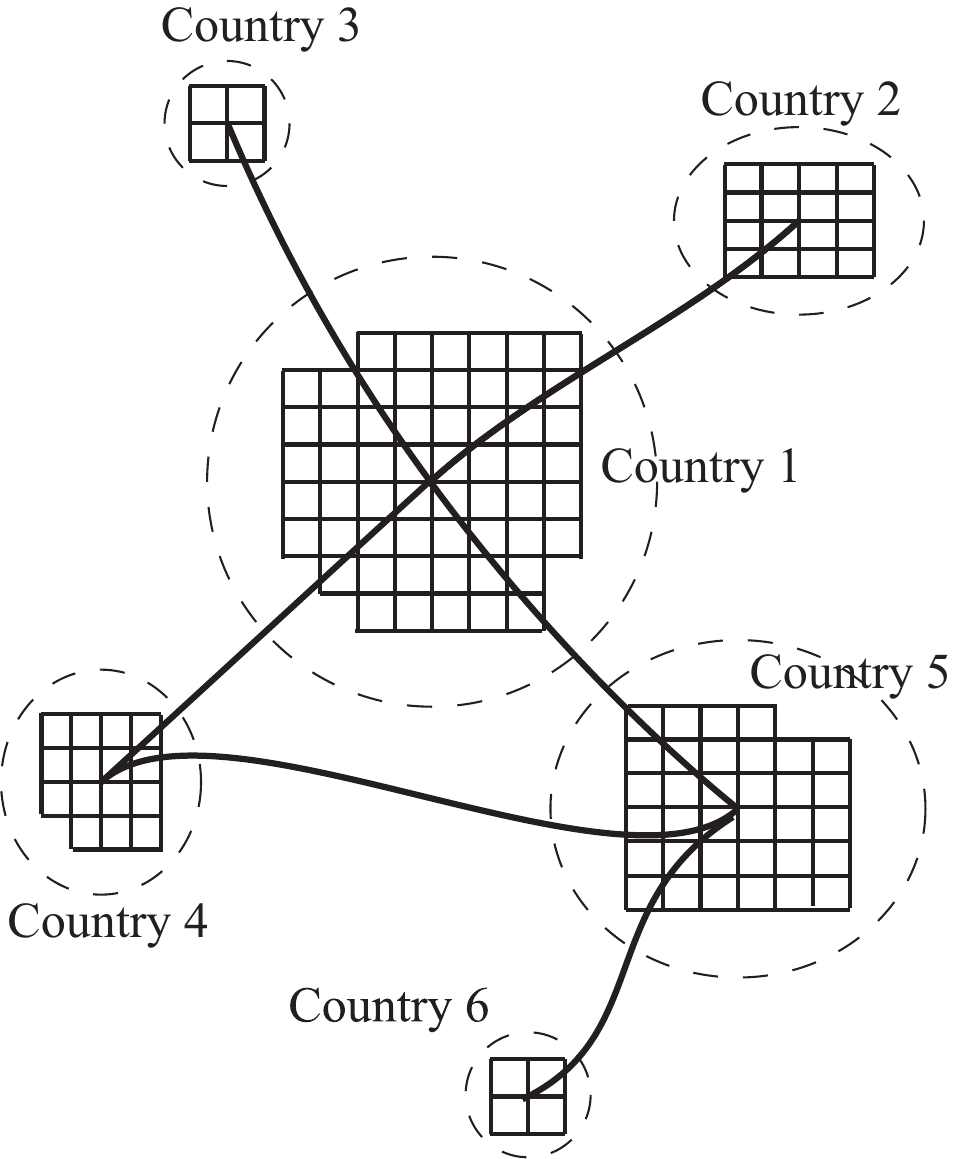} \\
(a) Continuous network &
(b) Discrete network 
\end{tabular}
   \caption{Global--local system comprising countries
   with local regions. 
   \label{Global--local system}}

\vspace{5mm}
   \centering\small

\vspace{2mm}
\includegraphics[scale=0.29]{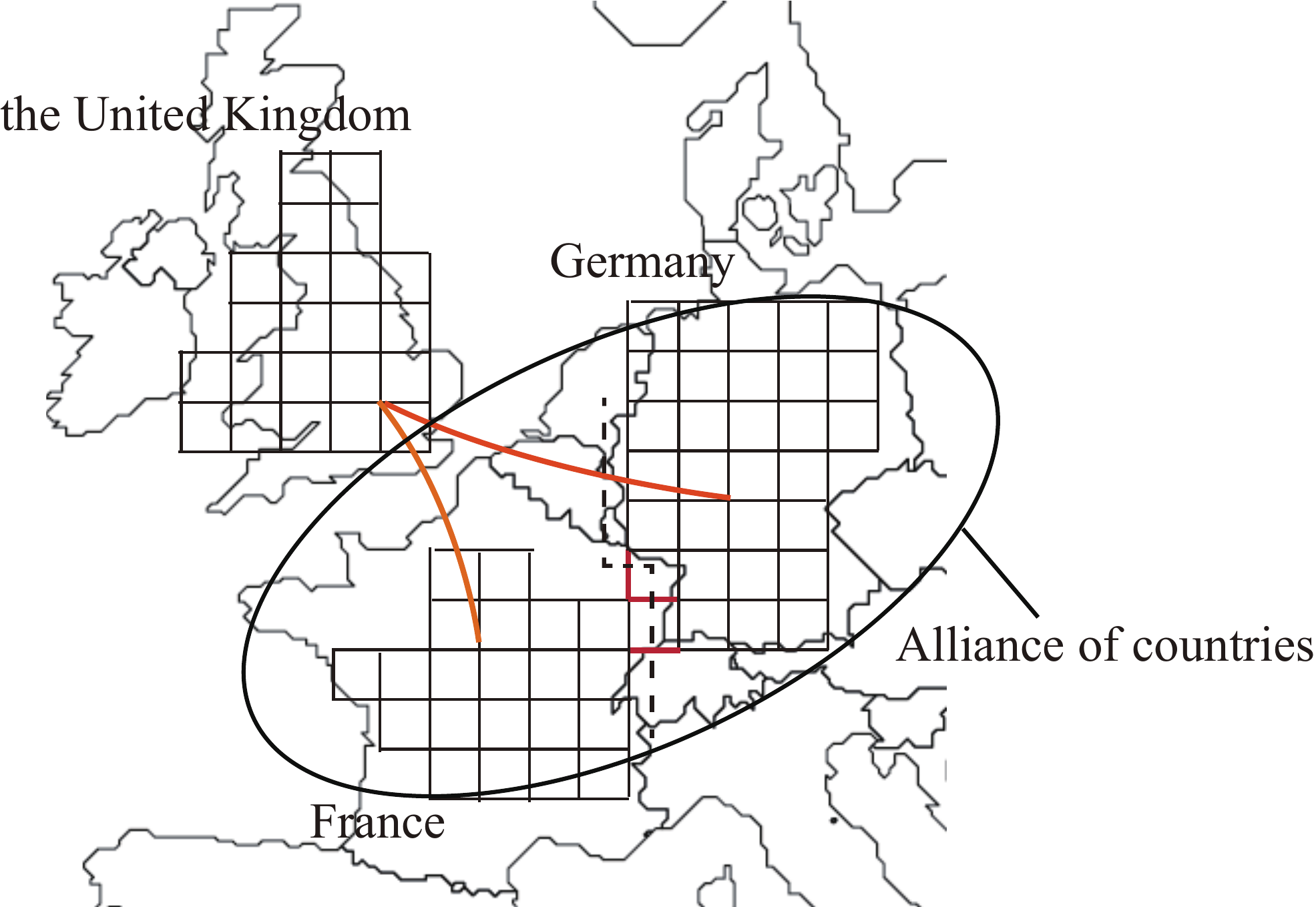}
   \caption{Spatial model of three countries:
the UK, France, and Germany.
Discrete local regions are located
at grid points of square lattices, and goods are transported along the 
lattice.
The number of local regions assigned to each country
reflects the relative sizes of its population and economy.
London, UK, is directly connected to
France (Paris) and Germany (Frankfurt),
as indicated by the red curves. 
France and Germany are continuously connected
by the red grid lines.
     \label{Global--local system 3 countries Intro}}
 \end{figure}

To generalize the conventional two-level hierarchical system of countries and regions commonly examined in economic geography,
this study employs a hierarchical spatial system,\footnote{
 An urban geographical theory 
(central place theory)
seeks to explain such a hierarchy of human settlements in a residential system
\citep{Christaller.1933.1966}.
} 
which, e.g., 
is structured as follows:
\begin{align} \label{multi-hierarchical Intro} &
 & \mbox{Regions} ~\rightarrow~ 
 \mbox{Countries} ~\rightarrow~
\mbox{Alliances of countries}.
\end{align}
There are several types of modeling of the hierarchy of places in \eqref{multi-hierarchical Intro}.\footnote{%
Recent studies have employed various spatial models of economic activity.
\citet{Redding.2016} and
\citet{Redding.Rossi-Hansberg.2017Ch01} use a latitude–longitude grid
of several sizes divided into two countries.
\citet{Allen.Arkolakis.2022} employ an irregular mesh to simulate the US highway network and a regular mesh to model Seattle's road network.
\citet{Fajgelbaum.Schaal.2020} use a $15 \times 15$ square grid and discretized road network models of France, Spain, and Western Europe.
\citet{Ikeda.Murota.2014Ch01} employ a hexagonal lattice to demonstrate the 
self-organization of hexagonal patterns in mobile population distributions.
} 
In this study,
we connect local networks
through two types of global-networks:\footnote{%
See \citet{Kogure.Ikeda.IJBC.2022Ch11} 
and \citet{{Ikeda.Takayama.2024}} for 
the analysis of 
a global--local system.
} 
Figure~\ref{Global--local system}(a) depicts a continuous network,
where distributed local regions are grouped into several countries.
The discrete network, shown in Fig.~\ref{Global--local system}(b),
can represent 
the US highway system as well as
high-speed railway networks and air transportation networks.
The centripetal force prevails
around the trade/transportation hubs
in this  discrete network.
Figure~\ref{Global--local system 3 countries Intro}
depicts a model of
the UK, France, and Germany.\footnote{This
 model combines both continuous and discrete networks: 
France and Germany are connected continuously, whereas the UK is linked to them
via a discrete network.
France and Germany can be treated as an alliance 
of countries within the EU single market.
} 
This model has a three-level hierarchy,
comprising regions, countries, and the alliance of countries.

Direct analysis of the governing equation of 
a hierarchical spatial system 
in a large-scale setting 
is computationally intensive and requires substantial data processing to derive economic insights.
To address this issue,
this study proposes a reduction analysis of the
 hierarchical system.
The governing equation of this system is reduced to
governing equations for higher-level places, such as
countries and alliances of countries. 
In the spatial model of the three countries
in Fig.~\ref{Global--local system 3 countries Intro},
the original region-level governing equation
with 109 degrees of freedom is reduced to the 
country-level and alliance-level equations with only
three and two degrees of freedom, respectively. 

This study aims to present systematic methods
for analyzing cross-country mobility in a large-scale hierarchical
system of places.
Based on the reduced governing equation,
we formulate 
a \textit{population-gradient matrix} to quantify the sensitivity of each country’s population to economic parameters, such as transportation costs.
The sensitivity of general economic parameters
to transport costs is formulated.

The proposed analysis is applied to
international trade competition
within the framework of economic geography models.
We consider internationally mobile workers,
given the progress of global trade liberalization, 
particularly within the EU.
A country is regarded as a winner (loser) when it experiences population gains (losses).
Tariffs are modeled as international trade costs, 
whereas reductions in domestic transportation costs represent
each country’s infrastructure development.
We analyze how changes in tariffs 
and domestic transportation costs
affect population distributions across and within 
countries.\footnote{\label{Behrens-footnote}%
\citet{Behrens.etal.2006,Behrens.etal.2007} 
present a seminal analysis of
the effects of both national transportation costs 
and international trade costs
on domestic population distributions in a two-country, four-region model.
They show that when international trade costs are low, population agglomeration occurs in one region of a country.
In contrast, this paper presents an innovative analysis
 method that can handle the spatial hierarchy in \eqref{multi-hierarchical Intro}
 in a large-scale setting.}

As a concrete example,
we highlight international trade competition in the context of Brexit
using the geographical model of the UK, France, and Germany
shown in Fig.~\ref{Global--local system 3 countries Intro}.
We use square lattices\footnote{%
The importance of the lattice analysis
of the combination of France and Germany
was suggested by Professor J.-F.~Thisse.
} 
to represent the countries’ local networks,
and employ, e.g.,\footnote{The analytical framework proposed here
is applicable to a broad class of 
economic models beyond the Helpman model.
} 
an economic geography 
model \`{a} la Helpman\footnote{\label{HelpmanImportance}%
This model is based on the multi-region version of the new economic geography model of \citet{Helpman.1998Ch01} \citep{Redding.Sturm.2008}.
  This multi-region version is similar to several quantitative spatial models (e.g., \citealp{Allen.Arkolakis.2014Ch01}; \citealp{Becker.Heblich.Sturm.2021}).}
with replicator dynamics.\footnote{%
Many economic geography models 
employ the dynamics 
together with the corresponding static governing equation.
For example, \textit{replicator and logit dynamics}
are used to investigate the stable equilibria of an economic system. 
} 
We examine how the three countries can design trade policies and develop national transport infrastructure to attract mobile workers. 
We find that, 
in the post-Brexit trade environment,
the UK gains population through
its domestic infrastructure development.
Trade liberalization benefits the UK irrespective of the EU’s tariff policy.
In the EU’s post-Brexit trade strategy,
investing in its domestic infrastructure favors the EU.
However, changes in import trade freeness may exert 
either positive or negative effects, 
depending on the tariff type.

This paper is structured as follows.
The next section reviews 
recent trade environment and related studies.
Section~\ref{EconomicModel}
introduces the 
economic geography model.
A reduction method for a hierarchical spatial system
is presented in Section \ref{StaticReduction}.
Section~\ref{ParameterSensitivity}
formulates parameter sensitivity analysis
for economic geography models.
The spatial and economic model of the UK, France, and Germany
is presented, and their joint economic integration
is analyzed in Section~\ref{Analysis of Three Countries}.
Trade strategies for the UK and for the EU
are examined in Sections~\ref{UK-Strategy}
and \ref{EU-Strategy}, respectively.
Section \ref{sec:global_concl} concludes.


\section{Recent Trade Environment and Related Studies}\label{Related Studies}

Average tariff rates in France, the UK, and the US have declined over time 
until recently (1790–2019; \citealp{Irwin.2020}).
Brexit in 2020 cast a shadow over the EU’s global economic integration 
and contributed to a decline in net migration from EU countries to the UK 
\citep{DiIasio.Wahba.2023}.
At present, globalization faces a threat of the resurgence of higher tariffs,
 particularly in the US.
 Changes in a country’s tariff rate often raise serious concerns
among other countries,
as evidenced by recent international trade tensions.
There is an urgent need to assess the impacts of these tariffs and to identify potential winners and losers in international trade competition.

Several studies 
have examined the effects of
uncertainty shocks generated by the 2016 Brexit referendum on labor markets, investment, and trade, as reviewed by \citet{Dhingra.Sampson.2022}.
\citet{Brakman.etal.2023} note: 
 ^^ ^^ As part of its ‘Global Britain’ strategy and as a reaction to 
 its Brexit decision, the UK government is pursuing a series of Free Trade Agreements with countries around the world, $\cdots$.''
\citet{DeLucio.etal.2024} demonstrate that the UK's withdrawal
from the EU reduced
 both Spanish exports to, and imports from, the UK.
\citet{Freeman.etal.2025} find that both UK imports and exports declined 
following the UK's withdrawal.
Among studies focusing on economic disintegration, such as Brexit, 
several studies examine its effects on the spatial distribution of 
workers and firms across countries.
\citet{Commendatore.Kubin.Sushko2} and \citet{Saraiva.Gaspar.2026}
explore how economic disintegration 
affects the spatial distribution of workers 
in a three-country model.
 \citet{Janeba.Schulz.2024} investigate
 how economic disintegration influences firm relocation and national tax policies 
 using a general-equilibrium trade model.

 The cross-country mobility of population and capital has been studied extensively.
\citet{Zeng.Zhao.2010}
theoretically examine how reductions in trade costs affect the spatial distribution of capital across countries in a two-country, four-region model.
They show that whether reductions in trade costs increase the amount of capital in a region depends on the country's population and its
 level of domestic transportation costs.
\citet{Persyn.etal.2023}
quantitatively examine how reduced transportation
costs influence population distribution in the EU and the UK, using a spatial dynamic general-equilibrium model across 267 European NUTS-2 regions.
They investigate the long-run impact of road-transport infrastructure investments under the 2014–2020 European Cohesion Policy and show that the policy can increase the population of targeted regions 
by up to 0.167\%.

An extensive body of research has emerged in quantitative spatial economics (QSE), as reviewed by \citet{Redding.Rossi-Hansberg.2017Ch01}
and \citet{Allen.Arkolakis.2025}.
Key contributions include the following:
\citet{Eaton.Kortum.2002} model international trade.
\citet{Allen.Arkolakis.2014Ch01} estimate the topography of trade costs, productivity, and amenities in the US using an irregular lattice.
\citet{Ahlfeldt.etal.2015} develop a model of internal city structure and apply it to data from thousands of city blocks in Berlin.
\citet{Caliendo.Parro.2015}
quantify the trade and welfare effects of NAFTA.
\citet{Caliendo.etal.2019}
quantify the effect of the China shock on manufacturing employment in the US from 2000 to 2007.
\citet{Redding.2016} quantitatively explores the effect of reductions in transportation costs on population distributions using a 
$20 \times 20$ latitude–longitude grid.
\citet{Behrens.etal.2017} introduce a multi-city general equilibrium model to analyze the influence of spatial frictions.
\citet{Desmet.etal.2018} propose a dynamic theory of spatial growth 
that incorporates realistic geography.
\citet{Behrens.Murata.2021} derive spatial-equilibrium conditions in QSE.
\citet{Fajgelbaum.etal.2020} conduct a simulation of real wage impacts from the US and retaliatory tariffs in a general-equilibrium model of the US economy.

\section{Economic Geography Model}\label{EconomicModel}

We introduce the economic geography model
employed in this paper.

\subsection{Dynamics and Governing Equation}%
\label{DynamicsGoverningEquation}

Many economic geography models 
employ the dynamics $\frac{{\rm d}{\bm \lambda}}{{\rm d}t}
=\bm{F}({\bm \lambda},\bm{\tau})$
together with the associated static governing equation:
 \begin{align}\label{nonlinearGovEq}
 \bm{F}({\bm \lambda},\bm{\tau})={\bf 0},
\end{align}
where $\bm{\lambda}=(\lambda_i)
\in \mathbb{R}\sp{n}$ denotes the vector of independent variables,
$\bm{\tau}=(\tau_k)\in \mathbb{R}\sp{p}$ indicates 
the vector of economic variables, and 
$\bm{F}=(F_i)\in \mathbb{R}^n$
denotes a sufficiently smooth nonlinear function.
A typical example of $\tau_k$ is a transportation cost,
and $\lambda_i$ denotes the population of mobile workers.
Hereafter, $(\cdot)_i$
denotes a variable at place $i \in N=\{1,\ldots,n\}$,
and $(\cdot)_k$ denotes the $k$th economic variable.

We employ the dynamics that satisfy the relations 
 \begin{align} & \label{law for dynamics conservation law} 
 \sum_{i\in N} F_i(\bm{\lambda},\bm{\tau})=0,
 \quad
  \sum_{i\in N} \lambda_i =1,
\end{align}
which are used in the reduction to a simplex 
 in Section~\ref{FurtherReduction}.  
The governing equations of customary dynamics, 
such as replicator and logit dynamics
satisfy these relations.\footnote{
Many economic geography models employ such dynamics together with 
the corresponding static governing equation. In particular, the equilibrium 
population share in QSE models with Fréchet-distributed location 
preferences takes the same multinomial-logit form as the steady state of 
the logit dynamics \citep{Eaton.Kortum.2002, Allen.Arkolakis.2014Ch01, Allen.Arkolakis.2022}.
} 
This generality 
allows our framework to encompass both NEG-type models and 
quantitative spatial models with Fréchet-distributed preferences.

In the international trade analysis
(cf. Sections~\ref{Analysis of Three Countries}--\ref{EU-Strategy}),
we adopt the \textit{replicator dynamics}
(\citealp{TaylorJonker1978Ch01};
\citealp{Sandholm.2010Ch01}):
  \begin{align}\label{eq:model_replicatorCh1}
    F_i(\bm{\lambda},\bm{\tau})= (v_i(\bm{\lambda},\bm{\tau}) - \bar{v}(\bm{\lambda},\bm{\tau})) \lambda_i
\end{align}
\noindent
with the weighted-average utility
$\bar{v} = (\sum_{i\in N} \lambda_i v_i)/\sum_{i\in N} \lambda_i$.
Each worker in place $i$ is assigned an indirect utility $v_i$,
as defined by the economic model, and migrates
across places in search of higher utility.

\subsection{Transportation Cost and Trade Freeness}\label{IcebergTransportationCost}

We adopt the iceberg transportation cost:
when one unit of goods is shipped from place $i$
to place $j$,
only $1/\tau_{ij}$ arrives.
We consider 
\begin{itemize}
\item
 $\tau_\alpha$: the national transportation cost 
 parameter
 within country $\alpha$.
\item
 $\tau_{\alpha\rightarrow \beta}$: 
 the trade cost parameter of exporting goods from country $\alpha$ to country $\beta$.
\end{itemize}
If places $i$ and $j$ are located in the same
country $\alpha$,
there is no trade cost,
and the transportation cost is given by
$\tau_{ij} = \exp(L_{ij} \tau_\alpha)$,
where $L_{ij}$ denotes the road distance between them.
When places $i$ and $j$ belong to different countries $\alpha$ and $\beta$,
and place $i^*$ in country $\alpha$ is directly connected to place $j^*$ in country $\beta$, 
a trade cost
$\exp(\tau_{\alpha\rightarrow \beta})$ emerges,
and the total cost
 between places $i$ and $j$ is given by
\begin{align*} &
\tau_{ij}=\exp (L_{i \, i^*} \, \tau_\alpha
+ L_{j \, j^*} \, \tau_\beta + \tau_{\alpha\rightarrow \beta}).
\end{align*}

We consider \textit{trade freeness parameter}, 
which is inversely related to the transportation cost.
In general economic geography models,
this parameter is defined as
\begin{align*}
& 
\phi_k =\exp[-(\sigma-1)\tau_k]
\qquad 
(k=\alpha,\alpha\rightarrow \beta;~0<\phi_k<1).
\end{align*}
Here,
$\sigma~(>1)$ is an economic parameter
representing the 
constant elasticity of substitution
(cf. Section~\ref{HelpmanModel}).

\subsection{The Helpman Model}\label{HelpmanModel}

We employ the multi-place version of the \citet{Helpman.1998Ch01} model
to analyze international trade competition.\footnote{%
  This model is closely related to several quantitative spatial models,
  such as those developed by \citet{Allen.Arkolakis.2014Ch01}
  and \citet{Becker.Heblich.Sturm.2021}.
} 
We outline this model below
and provide its details in \ref{HelpmanDetails}.

The economy is composed of two sectors:
housing and manufacturing.
Each place has a fixed stock of housing
that cannot be traded across places.
The manufacturing sector produces differentiated goods 
under monopolistic competition
with increasing returns to scale.
These differentiated goods can be traded
but are subject to transport costs across places.
Each place hosts a continuum of firms, 
and each firm produces a single type of differentiated good 
using only labor as the factor of production.
The labor market is perfectly competitive, 
and all firms treat wages as given.

The utility of a worker in place $i$ 
is expressed as
\begin{align}\label{utility}
  u_i = \left(\frac{Q_i}{\mu}\right)^{\mu}
  \left(\frac{h_i}{1 - \mu}\right)^{1 - \mu} 
  \qquad (0 < \mu < 1)
\end{align} 
in terms of the consumption index  $Q_i$
over differentiated traded goods,
the consumption of housing services $h_i$,
and the expenditure share $\mu$ allocated to the consumption of
differentiated goods.
The consumption index $Q_i$ is given by
\begin{align*} &
Q_i =  
\left(
    \sum_{j\in N} 
    \int_{0}^{m_j} q_{ji}(\varphi)^{(\sigma - 1) / \sigma} 
    \,\mathrm{d}\varphi
\right)^{\sigma / (\sigma - 1)}
\end{align*}
in terms of a constant elasticity of substitution (CES) aggregator
with elasticity $\sigma > 1$ between traded goods. 
Here, $m_j$ is the mass of varieties in place $j$ and
$q_{ji}(\varphi)$ is the consumption of the $\varphi$th 
differentiated good produced in place $j$ and consumed in place $i$.

The budget constraint of a worker in place $i$ is 
\begin{align*}
    \left( 
        \sum_{j\in N} \int_{0}^{m_j} p_{ji}(\varphi)q_{ji}(\varphi) \,\mathrm{d}\varphi
    \right) + 
    r_i h_i = Y_i .
\end{align*}
Here,
 $r_i$ is the housing price and
 $p_{ji}(\varphi)$ is the price of the $\varphi$th differentiated good
produced in place $j$ and consumed in place $i$.
Utility maximization yields the following:
\begin{align}\label{qijhivi}
    & q_{ji}(\varphi) = \frac{\mu Y_i}{p_{ji}(\varphi)} 
    \left(
        \frac{p_{ji}(\varphi)}{P_i} 
    \right)^{1 - \sigma}, \quad
    h_i = \frac{(1 - \mu) Y_i}{r_i}, \quad
    v_i = \frac{Y_i}{P_i^\mu r_i^{1 - \mu}}.
\end{align}
Here, $P_i$ is the price index: 
\begin{align}\label{price_index}
    P_i = 
    \left(
        \sum_{j\in N} \int_{0}^{m_j} p_{ji}(\varphi)^{1 - \sigma} \, \mathrm{d}\varphi
    \right)^{1 / (1 - \sigma)}.
\end{align}

Under utility maximization, profit maximization, and market clearing,
the indirect utility is given by
\begin{align}\label{closed_indirect_utility}
  v_i = \zeta
  \lambda_{i}^{\mu - 1} Y_i^\mu 
  \left(
    \sum_{j \in N} \lambda_j (\tau_{ji} w_j)^{1 - \sigma}
  \right )^{\mu / (\sigma - 1)},
\end{align}
where $\zeta$ is a constant that depends on exogenous parameters
and $Y_i$ is 
each worker's income in place $i$, which consists of wage income $w_i$
and income from housing services.  

The static spatial distribution of workers
$\bm{\lambda} = (\lambda_i)$
is determined by the governing equation in \eqref{nonlinearGovEq},
together with the replicator dynamics \eqref{eq:model_replicatorCh1}
and the indirect utility \eqref{closed_indirect_utility}.


\section{Reduction Method for Hierarchical Spatial System}\label{StaticReduction}

We consider a hierarchical spatial system that is
organized from lower to higher levels of places.
 As an example of this system, we employ
\begin{align} 
\label{multi-hierarchical urban} &
 & \mbox{Regions} ~\rightarrow~ 
 \mbox{Countries} ~\rightarrow~
\mbox{Alliances of countries},
\end{align}
while the proposed theoretical framework is 
extendable to more general hierarchies.
We propose a systematic method
for analyzing this hierarchical system within the framework of 
 economic geography.
We formulate the governing equation at a particular level
in the hierarchical system,
and derive information at any higher level 
using this method.

\subsection{Fundamental and Target Levels}

We consider
$n$ lower-level places,
referred to the \textit{fundamental level}, 
and $n^{\rm L}$ higher-level ones $(n>n^{\rm L})$,
referred to the \textit{target level} L.
In the hierarchy in \eqref{multi-hierarchical urban}, 
the  fundamental level corresponds to regions,
while 
the target level to countries or alliances of countries,
which we express ${\rm L}={\rm C}$ or A.

We aim to express target-level vectors $\bm{\lambda}^{\rm L}$ and $\bm{F}^{\rm L}$ in terms of
the fundamental-level vectors $\bm{\lambda}$ and $\bm{F}$.
We index fundamental-level places by integers $i\in N=\{1,\ldots,n\}$.
Each target-level place $\alpha$ comprises 
several fundamental-level places indexed by a subset
$\mathcal{C}_\alpha$ of $N$
$(N=\bigcup_{\alpha=1}^{n^{\rm L}}\mathcal{C}_\alpha)$.
The population component $\lambda_\alpha^{\rm L}$
of the place $\alpha$ is obtained by
summing the populations of the 
fundamental-level places belonging to $\alpha$ as
$\lambda_\alpha^{\rm L}=\sum_{i\in \mathcal{C}_\alpha}\lambda_i$.
Similarly, 
we define the governing equation component
as $F_\alpha^{\rm L}=\sum_{i\in \mathcal{C}_\alpha}F_i$. 
Then, we assemble these components
$\lambda_\alpha^{\rm L}$ and $F_\alpha^{\rm L}$ 
to define the target-level vectors as
\[
\bm{\lambda}^{\rm L}=(\lambda_\alpha^{\rm L}
=\sum_{i\in \mathcal{C}_\alpha}\lambda_i 
\mid 
\alpha=1,\ldots,n^{\rm L}), \quad
\bm{F}^{\rm {\rm L}}=(F_\alpha^{{\rm L}}
=\sum_{i\in \mathcal{C}_\alpha}F_i 
\mid \alpha = 1,\ldots,n^{\rm {\rm L}}).
\]
The target-level component $F_\alpha^{\rm L}$ 
has economic significance, as explained in Remark~\ref{Ai Implication}.

\begin{remark}\label{Ai Implication}
For the replicator dynamics in \eqref{eq:model_replicatorCh1},
we have
\begin{align*}
F_\alpha^{\rm L} &
= \sum_{i\in \mathcal{C}_\alpha}(v_i-\bar{v})\lambda_i
=\sum_{i\in \mathcal{C}_\alpha}v_i\lambda_i
 -\bar{v}\sum_{i\in \mathcal{C}_\alpha}\lambda_i= 
 \left(\bar{v}_\alpha-\bar{v}\right)\lambda_\alpha^{\rm L}
 \cr
 & =  \{(\mbox{weighted average utility $\bar{v}_\alpha$ of place $\alpha$})
 \\ & \hspace{10mm}
  -(\mbox{weighted average utility $\bar{v}$ of the whole world}) \}
 \\ & \hspace{55mm}
 \times (\mbox{place $\alpha$'s population $\lambda_\alpha^{\rm L}$}) ,
\end{align*}
where 
$\bar{v}_\alpha=(\sum_{i\in \mathcal{C}_\alpha}v_i\lambda_i)/\lambda_\alpha^{\rm L}$.
Thus, any target-level equation $F_\alpha^{\rm L}$
inherits the form of replicator dynamics of the fundamental level.
\hfill$\Box$
\end{remark}

\subsection{Reduction of Governing Equation}

As a preliminary step for the reduction analysis,
we derive the incremental governing equation.
By considering two solutions
$({\bm \lambda},\bm{\tau})$ and 
$({\bm \lambda}+{\rm d}{\bm \lambda}, 
\bm{\tau}+{\rm d} \bm{\tau})$
 of \eqref{nonlinearGovEq},
we obtain a fundamental-level incremental governing equation
\begin{align}\label{IncNonlinear}
{\rm d} \bm{F}({\bm \lambda},\bm{\tau}) & \equiv
 \bm{F}({\bm \lambda}+{\rm d}{\bm \lambda},
\bm{\tau}+{\rm d} \bm{\tau}) - \bm{F}({\bm \lambda},\bm{\tau})
\nonumber \\
& = J({\bm \lambda},\bm{\tau}){\rm d}{\bm \lambda}
 +G({\bm \lambda},\bm{\tau}){\rm d} \bm{\tau} +{\rm h.o.t.}=
 \mbox{$\bm{0}$}
\end{align}
for infinitesimal
increments ${\rm d}{\bm \lambda}$ and ${\rm d} \bm{\tau}$.
Here, $J={\partial {\bm F}}/{\partial {\bm \lambda}}$ and
$G=\partial {\bm F}/\partial \bm{\tau}$ are Jacobian matrices,
and h.o.t.~denotes higher order terms.
The form of $J$ depends on the dynamics under consideration,
while the form of $G$ depends also on the economic modeling.
In the replicator dynamics in
\eqref{eq:model_replicatorCh1}, 
we have
\begin{align} \label{GJacobianGeneral}
G=\frac{\partial \bm{F}}{\partial \bm{\tau}}
&=
\mathrm{diag}[\bm{\lambda}]\,
\left(
  \left(
    \frac{\partial \bm{v}}{\partial \bm{\tau}}
  \right)
-
\bm{\lambda}^\top
\left(
  \frac{\partial \bm{v}}{\partial \bm{\tau}}
\right) \bm{1}
\right),
\end{align}
where $\bm{1} = ({1, \ldots, 1})^{\top} \in \mathbb{R}\sp{n}$.
The explicit form of $\partial \bm{v} / \partial \bm{\tau}$ in the Helpman model is given in \eqref{par_equ_indi}.

We construct a transformation
from the fundamental-level
vectors $\bm{\lambda}$ and $\bm{F}$ to 
target-level ones
$\bm{\lambda}^{\rm L}$ and $\bm{F}^{\rm L}$ 
as\footnote{Since 
$\tilde{H}$ is invertible, $\bm{F} = \bm{0}$
 is equivalent to $\bm{F}^{\rm L}= \bm{0}$
  and $\bm{B} = \bm{0}$.
} 
\begin{align} \label{coordinate transformation} &
{\bm \lambda}=H\begin{pmatrix}
{\bm \lambda}^{\rm L} \cr
{\bm b}
\end{pmatrix}=
(H_a,H_b)\begin{pmatrix}
{\bm \lambda}^{\rm L} \cr
{\bm b}
\end{pmatrix}
=H_a {\bm \lambda}^{\rm L}+H_b{\bm b},
~~
\begin{pmatrix}
{\bm A} \cr
{\bm B}
\end{pmatrix}=\tilde{H}^\top \bm{F}=
\begin{pmatrix}
\tilde{H}_a^\top \bm{F} \cr
H_b^\top \bm{F}
\end{pmatrix} ,~~
\end{align}
using transformation matrices 
$H=(H_a,H_b)$ and $\tilde{H}=(\tilde{H}_a,H_b)$,
given in \ref{TransMatH},
and auxiliary vectors
  $\bm{b}$ and $\bm{B}$ of 
  dimension ($n-n^{\rm L}$).
By \eqref{coordinate transformation},
we can transform the Jacobian matrices $J$ 
and $G$, respectively, to 
\begin{align*} &
\tilde{H}^\top JH=
\begin{pmatrix}
J_{a} & J_{ab} \cr
J_{ba} & J_{b} 
\end{pmatrix},
\quad
\tilde{H}^\top G=
\begin{pmatrix}
G_a \cr
G_b
\end{pmatrix}.
\end{align*}

Using this transformation,
we can reduce the fundamental equation in \eqref{IncNonlinear} 
to the target-level equation
presented in the following lemma.\footnote{%
This elimination process is known as Lyapunov--Schmidt reduction in nonlinear mathematics  
 and is used
for a different purpose, namely, the bifurcation analysis of symmetric systems
(e.g., \citealp{Ikeda.Murota.2019Ch01}).  
} 

\begin{lemma}\label{country-levelCondense}
Under the condition that 
$J_{b}$ is nonsingular,
the target-level governing equation is given by
\begin{align} \label{Reduced eq general} &
 {\rm d} \bm{F}^{\rm L}({\bm \lambda}^{\rm L},\bm{\tau})  =
 J^{\rm L}{\rm d}{\bm \lambda}^{\rm L}
 +G^{\rm L}{\rm d}\bm{\tau} +{\rm h.o.t.}=
 \mbox{$\bm{0}$} 
\end{align}
with ${J}^{\rm L}=J_{a}-J_{ab}J_{b}^{-1} J_{ba}$ and
${G}^{\rm L}=G_a - J_{ab} J_{b}^{-1} G_b$.
\end{lemma}
\begin{proof}
See \ref{TransMatH} for the proof.
\end{proof}

  \subsection{Further Reduction to Simplex}\label{FurtherReduction}

We introduced above the reduction from $n$ lower-level places to $n^{\rm L}$ higher-level ones
$(n>n^{\rm L})$.
Certain dynamics---such as replicator and logit dynamics\allowbreak---permit
a further reduction to an $(n^{\rm L} - 1)$-dimensional simplex
by exploiting Lemma~\ref{Lemma reductions}.

\begin{lemma}\label{Lemma reductions}
For any dynamics that satisfies 
 $\sum_{i\in N} F_i(\bm{\lambda},\bm{\tau})=0$
  and $\sum_{i\in N} \lambda_i =1$
in \eqref{law for dynamics conservation law},
 the following
target-level relations hold.
\begin{align}
& 
\sum_{\alpha=1}^{n^{\rm L}} {\rm d}\lambda_\alpha^{\rm L}=0, \quad
\sum_{\alpha=1}^{n^{\rm L}} {\rm d}F_\alpha^{\rm L}=0; \qquad
\sum_{\alpha=1}^{n^{\rm L}} \frac{\partial F_\alpha^{\rm L}}{\partial \tau_k}=0
 \quad (k=1,\ldots,p).
\label{conservation m} 
\end{align}
\end{lemma}
\begin{proof}
The proof is given in \ref{Proof of Lemma reductions}.
\end{proof}

  \begin{figure}
   \centering\small
   \begin{tabular}{c@{\hspace{8mm}}c} 
 \includegraphics[scale=0.43]{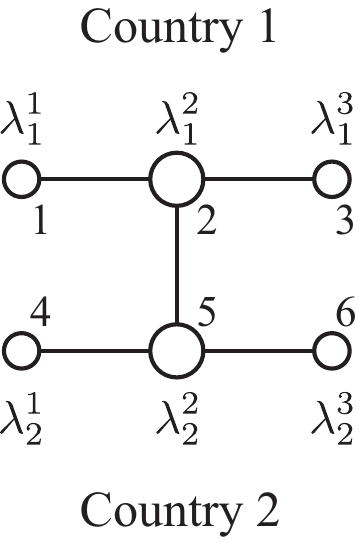} &
 \includegraphics[scale=0.43]{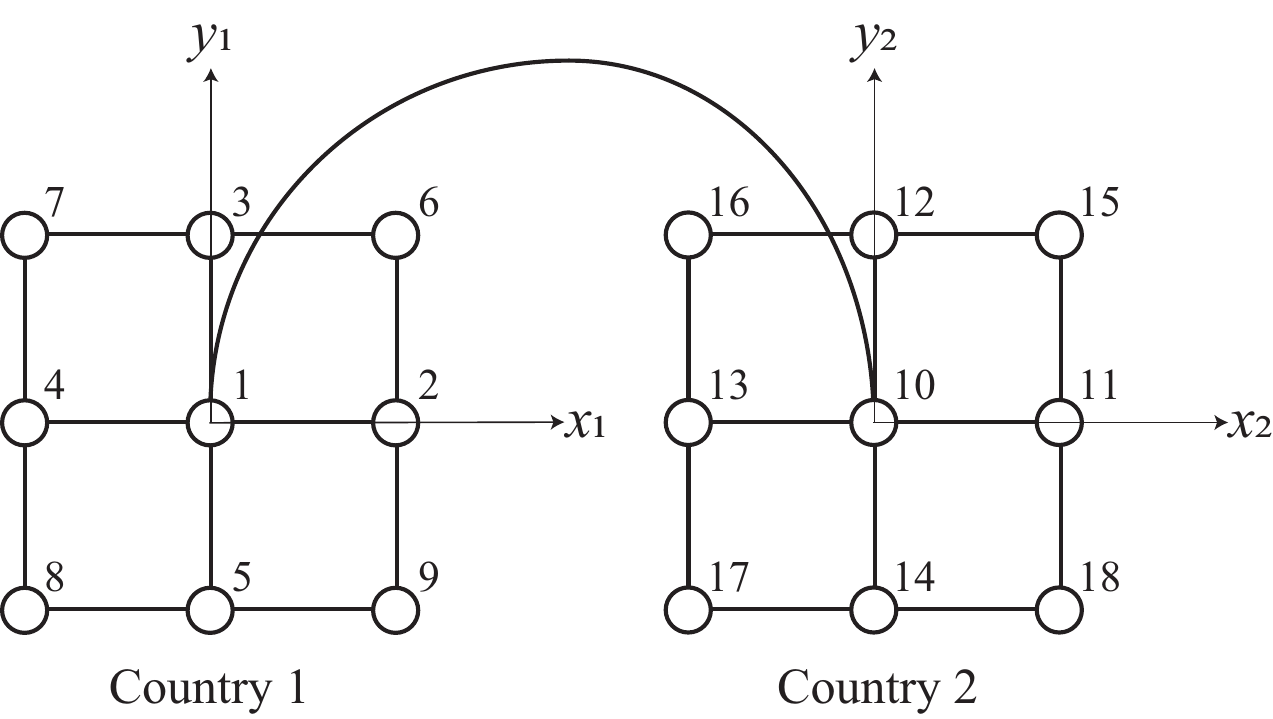} \\
(a) $3\times 2$ places &
(b) $9\times 2$ places
\end{tabular}
 
   \caption{A hierarchical spatial system consisting of two countries,
comprising several regions. 
     \label{2countries}}
 \end{figure}

We begin with a simple case of two countries ($n^{\rm L}=n^{\rm C}=2$)
and a single economic parameter $\bm{\tau}=\tau$
(cf.~Fig.~\ref{2countries}).
We suppress the superscript $(\cdot)^{\rm C}$ indicating
country-level variables.
The reduced two-dimensional system of equations for this case
is given by
\begin{align} \label{Reduced eq explicit} &
\begin{pmatrix}
{\rm d}F_1 \cr {\rm d}F_2
\end{pmatrix}
=
\begin{pmatrix}
J_{11} & J_{12} \cr
J_{21} & J_{22}
\end{pmatrix}
\begin{pmatrix}
{\rm d}\lambda_1 \cr {\rm d}\lambda_2
\end{pmatrix}
 +
\begin{pmatrix}
\frac{\partial F_1}{\partial \tau} \cr
\frac{\partial F_2}{\partial \tau}
\end{pmatrix}
{\rm d}\tau
+{\rm h.o.t.}={\bm 0}.
\end{align}
Using ${\rm d}\lambda_2 = -{\rm d}\lambda_1$
and
$\frac{\partial F_1}{\partial \tau} = -\frac{\partial F_2}{\partial \tau} 
\equiv g$,
which follow from \eqref{conservation m} with $n^{\rm L}=2$, and 
the projection matrix
$P=(1,-1)^\top$ 
(cf.~$n^{\rm L}=2$ in \eqref{PmatrixDefine}),
we obtain
\[
{\rm d}F_1-{\rm d}F_2=P^\top
\begin{pmatrix}
{\rm d}F_1 \cr {\rm d}F_2
\end{pmatrix},
\quad
\begin{pmatrix}
{\rm d}\lambda_1 \cr {\rm d}\lambda_2
\end{pmatrix}
=P{\rm d}\lambda_1,
\quad
\begin{pmatrix}
\frac{\partial F_1}{\partial \tau} \cr
\frac{\partial F_2}{\partial \tau}
\end{pmatrix}=Pg.
\]
Using these relations, we reduce the two-dimensional equation
in \eqref{Reduced eq explicit} 
to a one-dimensional equation
in terms of ${\rm d}\lambda_1$
(cf.~\ref{Proof 1D simplex}) as follows:
\begin{align} 
{\rm d}F_1-{\rm d}F_2  
    & =(J_{11} - J_{12} - J_{21} + J_{22}) \, {\rm d}\lambda_1
    +2g \, {\rm d}\tau
    +{\rm h.o.t.}=0.
    \label{1D simplex}
\end{align}

We next generalize the two-country case to
the $n^{\rm L}$-dimensional governing equation for target-level L
in \eqref{Reduced eq general}, 
as presented in Proposition~\ref{Prop for incremental}.

\begin{proposition}\label{Prop for incremental}
The governing equation 
in the $(n^{\rm L}-1)$-dimensional simplex is given by
\begin{align} \label{Gtau=Jd} &
  \tilde{J} {\rm d}\tilde{\bm{\lambda}}^{\rm L}+
\tilde{G}{\rm d}\bm{\tau}
+{\rm h.o.t.}={\bm 0}
\end{align}
with a vector
${\rm d} \tilde{\bm{\lambda}}^{\rm L}
=({\rm d}\lambda_1, \ldots, {\rm d}\lambda_{n^{\rm L}-1})^\top$,
an $(n^{\rm L}-1)\times (n^{\rm L}-1)$
matrix $\tilde{J}$, and
an $(n^{\rm L}-1)\times p$ matrix
$\tilde{G}$.
 \end{proposition}
\begin{proof}
The proof and the expressions of $\tilde{J}$ and $\tilde{G}$
are presented in \ref{Proof of Prop for incremental}.\end{proof}

\OMIT{
\subsection{Inverse Analysis}

In economic geography models,
population distributions are typically derived
for fixed economic parameters.
In contrast, 
we present a multiparameter inverse analysis
that identifies a specific set of economic parameters 
${\bm \tau} = {\bm \tau}^*$
that realizes a targeted
population distribution 
${\bm \lambda}^{\rm L} = ({\bm \lambda}^{\rm L})^{\rm *}$.
Section~\ref{UK'sContourMap} presents the application of this analysis
and \ref{InverseDetails} presents the details of the numerical inverse analysis.

In the case of $p=n^{\rm L}-1$
(see Remark~\ref{RemarkNon-Inverse} for other cases), 
the matrix
$\tilde{G}$
in \eqref{Gtau=Jd}
becomes a square matrix
and is generically invertible.
Then, \eqref{Gtau=Jd} can be solved as
\begin{align} \label{IncreReduced} &
{\rm d}\bm{\tau}=-\tilde{G}^{-1}\tilde{J} {\rm d}\tilde{\bm{\lambda}}^{\rm L}
+{\rm h.o.t.}
\end{align}
The inverse problem can then be formulated as follows:
\begin{align} \label{Problem Setting} 
\mbox{Find ${\bm \tau}={\bm \tau}^*$
such that ${\bm \lambda}^{\rm L}=
({\bm \lambda}^{\rm L})^*$, based on \eqref{IncreReduced}}.
\end{align}

\begin{remark}\label{RemarkNon-Inverse}
If $p < n^{\rm L} - 1$,
the inverse problem in \eqref{Problem Setting} cannot be solved exactly.
If $p > n^{\rm L} - 1$,
the number of economic parameters exceeds the number of equations,
yielding a set of non-unique solutions for $\bm{\tau}^*$.
Although $G$ is a square matrix if $p = n^{\rm L}$,
$G$ is singular and is not invertible due to the third 
relation in \eqref{conservation m};
accordingly, the reduction to the simplex 
in Proposition~\ref{Prop for incremental}
is mandatory.\hfill $\Box$
\end{remark}
}

\section{Parameter Sensitivity Analysis}\label{ParameterSensitivity}

We present a sensitivity analysis for the reduced system first
as a primary contribution of this paper.
Next, the direct analysis for the original system is presented
using the trade value as an example.

\subsection{Analysis of Reduced System by Population Gradient Matrix}
\label{ReductionSingleCountry}

In preparation for the sensitivity analysis
of the target-level governing equation.
This equation in Lemma \ref{country-levelCondense}
is solved for population increment 
${\rm d}{\bm \lambda}^{\rm L}$
as follows.

\begin{proposition}\label{RecursiveCondense}
When ${J}^{\rm L}$ is non-singular,
we have an explicit
${\rm d}{\bm \lambda}^{\rm L}$ versus
${\rm d}\bm{\tau}$ relation: 
\begin{align*} 
 {\rm d}{\bm \lambda}^{\rm L} & =
 T^{\rm L} {\rm d}\bm{\tau}
    +{\rm h.o.t.} 
\end{align*}
with a
\textit{population-gradient matrix}
 $T^{\rm L}=-({J}^{\rm L})^{-1}{G}^{\rm L}$.
 \end{proposition}

We hereafter set $\tau_k=\phi_k$ and 
denote by $t_{\alpha k}=\frac{\partial \lambda_\alpha^{\rm L}}{\partial \phi_k}$ the component of the matrix
$T^{\rm L}$, that is,
\begin{align} \label{GraDefine} &
T^{\rm L}=\left(t_{\alpha k}
=\frac{\partial \lambda_\alpha^{\rm L}}{\partial \phi_k}
\mid \alpha=1,\ldots,n^{\rm L};~k=1,\ldots,p\right).
\end{align}
We highlight this component
as a systematic tool for analyzing
international trade competition
(cf.~Sections~\ref{Analysis of Three Countries}--\ref{EU-Strategy}).
We can evaluate
the effects of trade-freeness parameters $\phi_k$
on the population distribution vector ${\bm \lambda}^{\rm L}$.
If $t_{\alpha k}>0$,
then the population $\lambda_\alpha^{\rm L}$ increases with $\phi_k$.
If $|t_{\alpha k}|$ is large, then the parameter $\phi_k$ is considered influential 
for the population $\lambda_\alpha^{\rm L}$.

\begin{definition}
When $\phi_k$ increases, place $\alpha$ is in  
\begin{equation} \label{WinnerLoser}
\left\{ \begin{array}{ll}
\mbox{strong position}
\quad & \mbox{if}~t_{\alpha k}>0, \cr 
\mbox{weak position} & \mbox{if}~t_{\alpha k}<0. \cr 
\end{array}\right.
\end{equation}
A target-level place $\alpha$ is in a \textit{globally strong} 
(or \textit{globally weak}) position
if $t_{\alpha k} > 0$ (or $t_{\alpha k} < 0$)
holds for any $\phi_k \in (0, 1)$.
\hfill$\Box$
\end{definition}

We have further conditions on commonly used dynamics.

\begin{proposition}
For dynamics that satisfies 
the conservation law of population $\sum_{i\in N} \lambda_i =1$
in \eqref{law for dynamics conservation law}, 
we have
\begin{align} \label{L conservation} &
\sum_{\alpha=1}^{n^{\rm L}} \lambda_\alpha^{\rm L} =1,
\quad
\sum_{\alpha=1}^{n^{\rm L}} t_{\alpha k}=0.
\end{align}
\end{proposition}
Generically,
at least one place is in a strong position ($t_{\alpha k} > 0$) and 
at least one place is in a weak position ($t_{\beta k} < 0$)
for this dynamics.
In the two-country case,
as $\phi_k$ increases, one country occupies
a strong position and the other a weak position;
when $\phi_k$ decreases, these positions are reversed.

\subsection{Direct Parameter Sensitivity Analysis}

We consider a solution $(\bm{\lambda}(\bm{\tau}),\bm{\tau})$
of the governing equation \eqref{nonlinearGovEq}
that is parameterized by $\bm{\tau}$.
Then, an economic variable, say $E$,
becomes a function $E^*=E(\bm{\lambda}(\bm{\tau}),\bm{\tau})$ in 
the solution space.
We have
\begin{align} \label{de/dtau} &
\frac{{\rm d} E^*}{{\rm d} \tau_k}
=
\frac{\partial E}{\partial \tau_k}
+
\sum_{i\in N}
\frac{\partial E}{\partial \lambda_i}
\frac{\partial \lambda_i}{\partial \tau_k}
\Longrightarrow
E_t^* \equiv
\frac{{\rm d} E}{{\rm d} \bm{\tau}}
=
\frac{\partial E}{\partial \bm{\tau}}
+
\frac{\partial E}{\partial \bm{\lambda}}
\frac{\partial \bm{\lambda}}{\partial \bm{\tau}}.
\end{align}
We call $E_t^*$ \textit{parameter sensitivity vector}.
Using the governing equation \eqref{IncNonlinear},
 we have
$\frac{\partial \bm{\lambda}}{\partial \bm{\tau}}=-J^{-1}G$
 and can rewrite this equation into
\begin{align} \label{EtOnEqui} &
E_t^* \equiv
\frac{\partial E}{\partial \bm{\tau}}
-\frac{\partial E}{\partial \bm{\lambda}}J^{-1}G.
\end{align}
\noindent
We can derive the sensitivity matrix
$E_t=\frac{\partial E}{\partial \bm{\lambda}^{\rm L}}$
with respect to  $\bm{\lambda}^{\rm L}$.

\OMIT{
\begin{proposition}\label{PropEtReduced}
\[
E_t=\frac{\partial E}{\partial \bm{\lambda}^{\rm L}}=
\frac{\partial E}{\partial \bm{\lambda}}
\frac{\partial \bm{\lambda}}{\partial \bm{\lambda}^{\rm L}}
=\frac{\partial E}{\partial \bm{\lambda}}
(H_a-H_b J_{b}^{-1} J_{ba}).
\]
\end{proposition}
\begin{proof}
The proof is given in \ref{ProofEtReduced}
\end{proof}
}

\subsection{Sensitivity of Trade Value to Transport Cost Parameter}

As an example, we consider 
the value of exports from place $i$ to place $j$
for the Helpman model.
This value is defined as 
\begin{align} &
V_{ij} = \int_0^{m_i}p_{ij}q_{ij} \lambda_j {\rm d} \varphi.
\end{align}
Since prices and demands do not vary across varieties, we have $V_{ij} = m_i p_{ij} q_{ij} \lambda_j$.
The change in the trade value $V_{ij}$ induced by a change in transportation cost $\tau_k$ can be decomposed into the effects of changes in price, the mass of varieties, demand, and population:
\begin{align}
  \frac{\mathrm{d} V_{ij}}{\mathrm{d} \tau_k}
  & = \Delta_p+\Delta_m+\Delta_q+\Delta_\lambda \nonumber \\
     & = 
  m_i q_{ij} \lambda_j \frac{\mathrm{d} p_{ij}}{\mathrm{d} \tau_k} + 
  p_{ij} q_{ij} \lambda_j \frac{\mathrm{d} m_i}{\mathrm{d} \tau_k} + 
  m_i p_{ij} \lambda_j \frac{\mathrm{d} q_{ij}}{\mathrm{d} \tau_k} + 
  m_i p_{ij} q_{ij} 
\frac{\mathrm{d} \lambda_j}{\mathrm{d} \tau_k}.
\label{decomposition_trade}\end{align}
The first term
$\Delta_p=m_i q_{ij} \lambda_j 
\frac{\mathrm{d} p_{ij}}{\mathrm{d} \tau_k}$
 represents the price effect that captures how the trade value changes through variations in the price $p_{ij}$ of goods shipped from place $i$ to place $j$, holding other variables constant.
The second term $\Delta_m$
represents the variety effect that reflects the impact of changes in the mass of varieties produced in place $i$, $m_i$, on the trade value from place $i$ to $j$.
The third term $\Delta_q$ represents the demand effect that measures how changes in the demand $q_{ij}$ for varieties from place $i$ in place $j$ affect the trade value.
The sum $\Delta_p+\Delta_m+\Delta_q$
of the first three terms
corresponds to 
$\frac{\partial E}{\partial \bm{\tau}}$
in the general expression \eqref{EtOnEqui},
and is in line with standard CES-based approaches to the interpretation of trade flows \citep[e.g.,][]{Redding.Venable.2004, Head.Mayer.2014}. 
The Helpman model further includes the fourth term $\Delta_\lambda$,
expressing the population effect that captures how changes in the population $\lambda_j$ in place $j$ affect the trade value by altering market size in place $j$.
We can use $\frac{\partial \bm{\lambda}}{\partial \bm{\tau}}=-J^{-1}G$ to express this term.


\section{The UK, France, and Germany in the EU Single Market}\label{Analysis of Three Countries}

By the proposed theoretical framework, we
analyze trade competition among the UK, France, and Germany
in the hierarchical spatial system
shown in Fig.~\ref{Global--local system 3 countries Intro},\footnote{%
A finer irregular mesh could be used to investigate local properties in greater detail.
} 
calibrated to the pre-Brexit period
(cf.~Section~\ref{Modeling and Trade Framework}).
This system comprises three levels: regions, countries and an alliance of countries.
We choose the regions as the fundamental level,
and set the target level as either
\[
\left\{
\begin{array}{ll}
\mbox{country level}: 
 & \mbox{the UK, France, and Germany}, \\
\mbox{alliance level}: 
 & \mbox{the UK and the EU}. \\
\end{array}
\right.
\]
Here, the EU refers to the alliance of France and Germany.

We employ the \citet{Helpman.1998Ch01} model 
(cf.~Section~\ref{HelpmanModel})
with replicator dynamics,
specifying its parameters 
as $\sigma=5.0$ and $\mu=0.75$,\footnote{%
The values of $\sigma$ and $\mu$ used in this study follow 
\citet[\S 3.9]{Redding.Rossi-Hansberg.2017Ch01}.
} 
while the analytical framework proposed here is applicable to a broad class of 
economic models beyond the Helpman model.

\subsection{Spatial Modeling and Trade Framework}\label{Modeling and Trade Framework}

The 2020 population ratios of the UK, France, and Germany are\footnote{%
See https://www.populationpyramid.net:
Population Pyramids of the World.
 } 
\begin{align}
 & 67,351,861 : 65,905,277 : 83,628,708 
  \ \approx \ 31.1 \% : 30.4 \% : 38.6 \% .
  \label{3 countries popu ratio}
\end{align}
Based on these ratios, the numbers of local regions representing
the UK, France, and Germany are set to 34, 33, and 42, respectively,
in Fig.~\ref{Global--local system 3 countries Intro}.

Domestic trade freeness for each country 
or the EU 
(designated as $\alpha$)
is defined as
\[
\phi_\alpha\quad (\alpha={\rm UK, Fra, Ger, EU}).
\]
International trade freeness
for exports from $\beta$ to $\alpha$ is defined as
\[
\phi_{\beta \rightarrow \alpha}
 \quad (\alpha,\beta={\rm UK},{\rm Fra},{\rm Ger},{\rm EU};~
\alpha\neq \beta).
\]
We employ two types of tariffs:
\begin{align} \label{TwoTariffTypes}  &
\left\{\begin{array}{ll}
\mbox{reciprocal tariff}: \quad &
\phi_{\beta \rightarrow \alpha}=\phi_{\alpha \rightarrow \beta}, \\
\mbox{asymmetric tariff}: \quad &
\phi_{\beta \rightarrow \alpha}\neq\phi_{\alpha \rightarrow \beta}. 
\end{array}\right.
\end{align}
A country (or alliance of countries) $\alpha$ is assumed to control
 its domestic trade freeness 
$\phi_{\alpha}$ and the import trade freeness
$\phi_{\beta \rightarrow \alpha}$.

France and Germany form an economic 
alliance represented by a single domestic trade-freeness parameter:
\begin{align} \label{SeamlessEU}
 & \phi_{\rm EU} \equiv \phi_{\rm Fra}
 =\phi_{\rm Ger}
 = \phi_{{\rm Fra}\rightarrow{\rm Ger}}
= \phi_{{\rm Ger}\rightarrow{\rm Fra}}
\end{align}
and apply the same level of tariffs 
to the trade with the UK,
that is,
\begin{align} \label{phiIntAsym}
 &  
 \phi_{\rightarrow {\rm UK}}\equiv
 \phi_{{\rm Fra}\rightarrow{\rm UK}}
  = \phi_{{\rm Ger}\rightarrow{\rm UK}},
 \quad
 \phi_{\rightarrow {\rm EU}}\equiv
 \phi_{{\rm UK}\rightarrow{\rm Fra}} 
 = \phi_{{\rm UK}\rightarrow{\rm Ger}}.
\end{align}

Using these trade-freeness parameters, we analyze several scenarios:
the pre-Brexit EU single market of the three countries 
 (Section~\ref{CooperativeEconomicIntegration}),
the post-Brexit trade strategy
of the UK in Section~\ref{UK-Strategy},
and that of the EU in Section~\ref{EU-Strategy}.

\subsection{Trade Analysis of the EU Single Market}%
\label{CooperativeEconomicIntegration}

We model the pre-Brexit EU single market
using the domestic trade freeness $\phi$
for the single EU market
and 
the international trade freeness $\phi_{\rm Int}$
under the reciprocal tariff
in \eqref{TwoTariffTypes}, being defined as
\begin{align*} &
\phi \equiv \phi_{\rm UK}=\phi_{\rm EU},
\quad
\phi_{\rm Int} \equiv 
 \phi_{\rightarrow{\rm  UK}} =
 \phi_{\rightarrow{\rm  EU}}.
 \end{align*}

Figure~\ref{PathTracing gamma=0.3UK}(a)
shows the $\phi_{\rm Int}$--$\lambda_{\rm UK}$ curve for $\phi=0.3$ 
and Panel (b) depicts
the associated  
region-level population distributions.
As trade liberalization progresses
(as $\phi_{\rm Int}$ increases),
agglomeration to a single city 
in each country is observed,\footnote{%
Agglomeration is observed
 around the hubs of direct international trade
(London, UK; Paris, France; Frankfurt, Germany).
As $\phi_{\rm Int}$ increases,
this agglomeration intensifies from points A to F,
benefiting the trading hubs. 
This tendency is in accordance with 
\citet{Behrens.etal.2006,Behrens.etal.2007}
(cf. Footnote~\ref{Behrens-footnote}).
} 

The country-level population distribution at point O,  
with $(\phi,\phi_{\rm Int})=(0.3, 0.3)$,
is
$\bm{\lambda}^{\rm C}=(\lambda_{\rm UK},\lambda_{\rm Fra},\lambda_{\rm Ger})
\approx (0.309, \allowbreak 0.299, 0.392)$
and is close to the ratios
in \eqref{3 countries popu ratio}
based on 2020 pre-Brexit population data.
Accordingly, we represent
the state of the EU single market by point O
and refer to it as the \textit{origin point}.

 \begin{figure}
   \centering\small
   \begin{tabular}{cc}
   \includegraphics[scale=0.4]{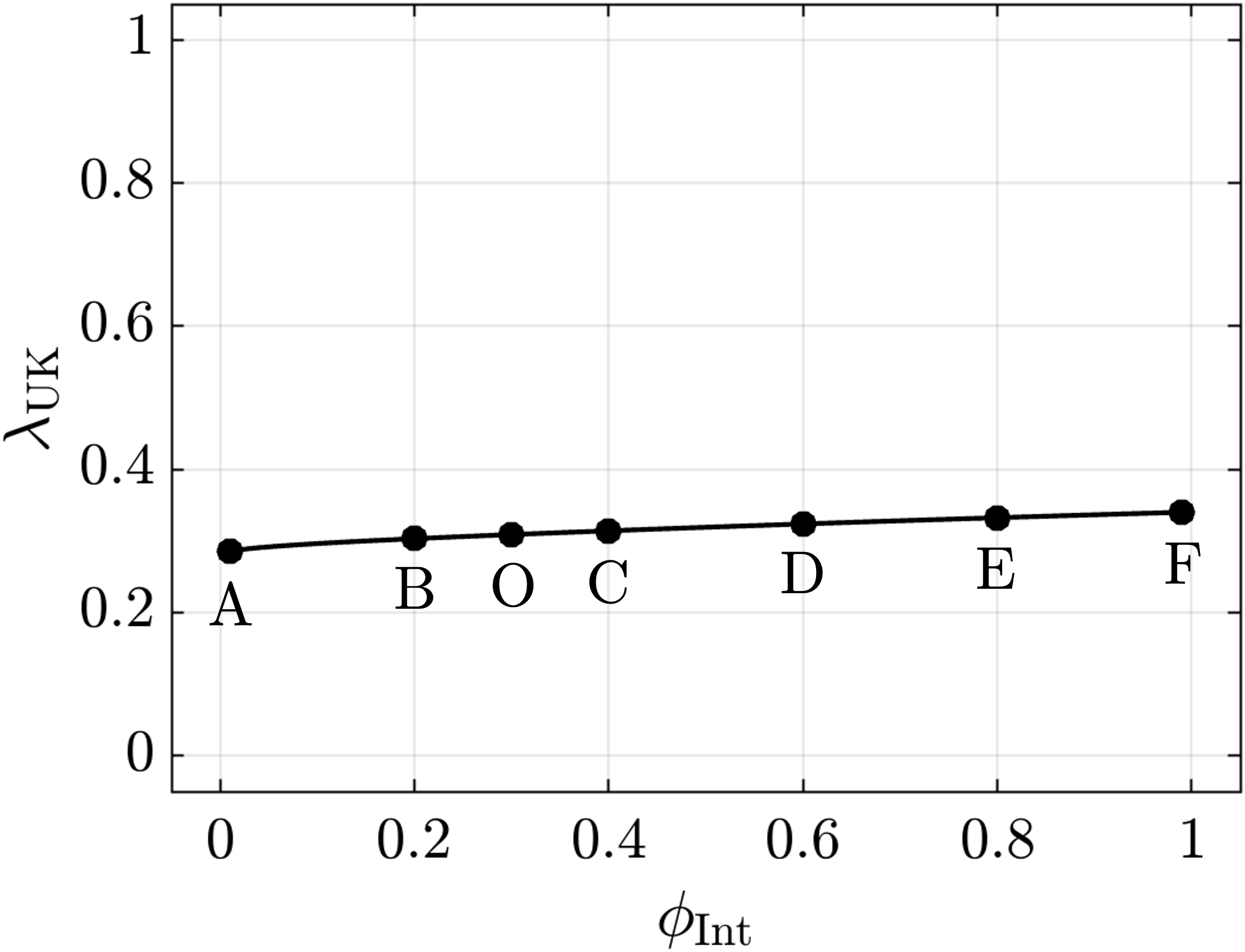} &
   \includegraphics[scale=0.34]{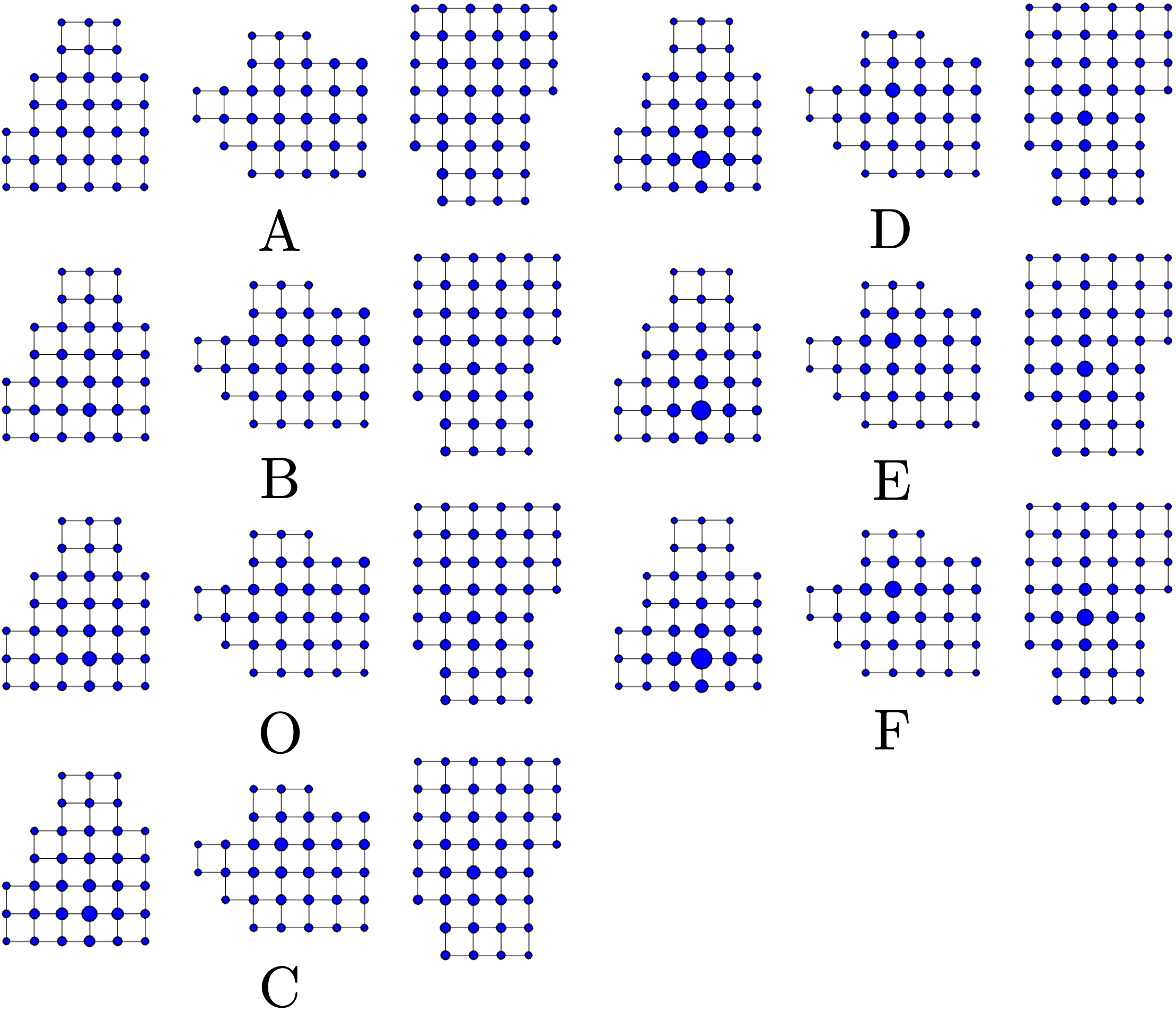} \\
      (a) $\phi_{\rm Int}$--$\lambda_{\rm UK}$ curve
   for $\phi=0.3$ \hspace{20mm}
   & (b) Population distributions 
   \\
   \end{tabular}
   \vspace{-1mm}
 \caption{$\lambda_{\rm UK}$ plotted against $\phi_{\rm Int}$ 
   for $\phi=0.3$.
   The $\phi_{\rm Int}$-$\lambda_{\rm UK}$ curve in (a)
    was obtained by comparative statics. The 
    region-level population distribution of the three countries
    at Points A to F are depicted in (b).}
   \label{PathTracing gamma=0.3UK}  
   \end{figure}

We conduct a local analysis at the origin point O,
using an alliance-level
population-gradient matrix 
(${\rm L}={\rm A}$ in \eqref{GraDefine}),
evaluated at this point 
for $\bm{\lambda}^{\rm A}=(\lambda_{\rm UK},\lambda_{\rm EU})$ and
$\bm{\tau}=(\phi_{\rm Int},\phi,\sigma,\mu)$ as
\begin{align*}
 & T^{\rm A}=
 \begin{pmatrix}
 t_{{\rm UK}, \phi_{\rm Int}} & \!
 t_{{\rm UK}, \phi} & \!
 t_{{\rm UK}, \sigma} & \!
 t_{{\rm UK}, \mu} \cr
 t_{{\rm EU}, \phi_{\rm Int}} & \!
 t_{{\rm EU}, \phi} & \!
 t_{{\rm EU}, \sigma} & \!
 t_{{\rm EU}, \mu} 
  \end{pmatrix}
 =
 \begin{small}
\begin{pmatrix}
\mbox{}+0.0546 \ \mbox{}-0.0209 \ \mbox{} -0.0000  \ \mbox{} +0.0055 \cr
\mbox{} -0.0546 \ \mbox{}+0.0209 \ \mbox{} +0.0000 \ \mbox{} -0.0055
 \end{pmatrix}.
 \end{small}
 \end{align*} 
The economic parameters
$\sigma$ and $\mu$ 
have components
of much smaller absolute values than
those of $\phi_{\rm Int}$ and $\phi$
and are therefore expected to be less influential.\footnote{%
We conducted the numerical analysis to ensure that
this is especially the case for $\sigma > 4$
 and $\mu < 0.8$. 
Our selection of $(\sigma,\mu)=(5.0,0.75)$
belongs to this case.
} 
Accordingly, we hereafter investigate the effects of
$\phi_{\rm Int}$ and $\phi$.

The UK's population ratio $\lambda_{\rm UK}$ 
increases with $\phi_{\rm Int}$,
because $ t_{{\rm UK}, \phi_{\rm Int}}=0.0546>0$.
Accordingly, the UK is in a strong position as
trade liberalization progresses (cf.~\eqref{WinnerLoser}).\footnote{%
The UK is in a 
globally strong position,
in accordance with the positive slope of the curve 
for all values of $\phi_{\rm Int}$
in Fig.~\ref{PathTracing gamma=0.3UK}(a).
} 
This result is somewhat ironic, since events such as Brexit
 (a decrease in $\phi_{\rm Int}$) undermine the UK.
In the three countries' joint domestic development (an increase in $\phi$),
the UK with $t_{{\rm UK}, \phi}=-0.0209<0$ is in a weak position,
whereas the EU with $t_{{\rm EU}, \phi}=0.0209>0$
is in a strong position
(see Section~\ref{ReductionSingleCountry}
for the reciprocity of the positions of two countries).
 These findings suggest that 
the UK and the other EU countries
 had stakes in each other's economic activities 
 in the EU single market prior to Brexit.

 We investigate the competition 
 between France and Germany within the EU, using
 the country-level population-gradient matrix 
 \begin{align*}
 & 
 T^{\rm C}=
 \begin{pmatrix}
 t_{{\rm UK}, \phi_{\rm Int}} &
 t_{{\rm UK}, \phi} \cr
 t_{{\rm Fra}, \phi_{\rm Int}} &
 t_{{\rm Fra}, \phi} \cr
 t_{{\rm Ger}, \phi_{\rm Int}} &
 t_{{\rm Ger}, \phi} 
 \end{pmatrix}
 =
 \begin{pmatrix}
+0.0546 & -0.0209 \cr
-0.0097 & -0.0031 \cr
-0.0449 & +0.0240 
 \end{pmatrix}
 \end{align*} 
for $\bm{\lambda}^{\rm C}=(\lambda_{\rm UK},\lambda_{\rm Fra},\lambda_{\rm Ger})$,
evaluated at the origin point O.
Since 
$t_{{\rm Fra}, \phi_{\rm Int}}=-0.0097<0$
and $t_{{\rm Ger}, \phi_{\rm Int}}=-0.0449<0$,
France and Germany are in weak positions under  
 trade liberalization (an increase in $\phi_{\rm Int}$),
 as is the EU.
 When $\phi$ increases,
 France is in a weak position
 ($t_{{\rm Fra}, \phi}=-0.0031<0$), 
 whereas 
 Germany is in a strong position
 ($t_{{\rm Ger}, \phi}=0.0240>0$).
Thus, France and Germany
have conflicting interests
within the EU single market
that must be carefully managed.
As shown above,
the population-gradient matrix $T$ helps analyze such conflicting interests.

\begin{remark}
As $\phi_{\rm Int}$ decreases from the 
original value of 0.3, modeling Brexit,
(i) the UK loses mobile population
(cf. Fig.~\ref{PathTracing gamma=0.3UK})
and (ii) the trade value declines sharply
(cf.~Fig.~\ref{TradeVolumeSce1}(a)).
The former aligns with the study by \citet{DiIasio.Wahba.2023},
which state: ^^ ^^ the UK has become less attractive
to EU potential and current immigrants.”
The latter agrees with \citet{DeLucio.etal.2024},
who report: ^^ ^^ We find that Spanish exports and imports to the UK decreased by 24\% and 27\%,  respectively, compared to the period before the Brexit referendum.”
\hfill$\Box$
\end{remark}

\begin{figure}
   \centering\small

\begin{tabular}{cc}
   \includegraphics[scale=0.37]{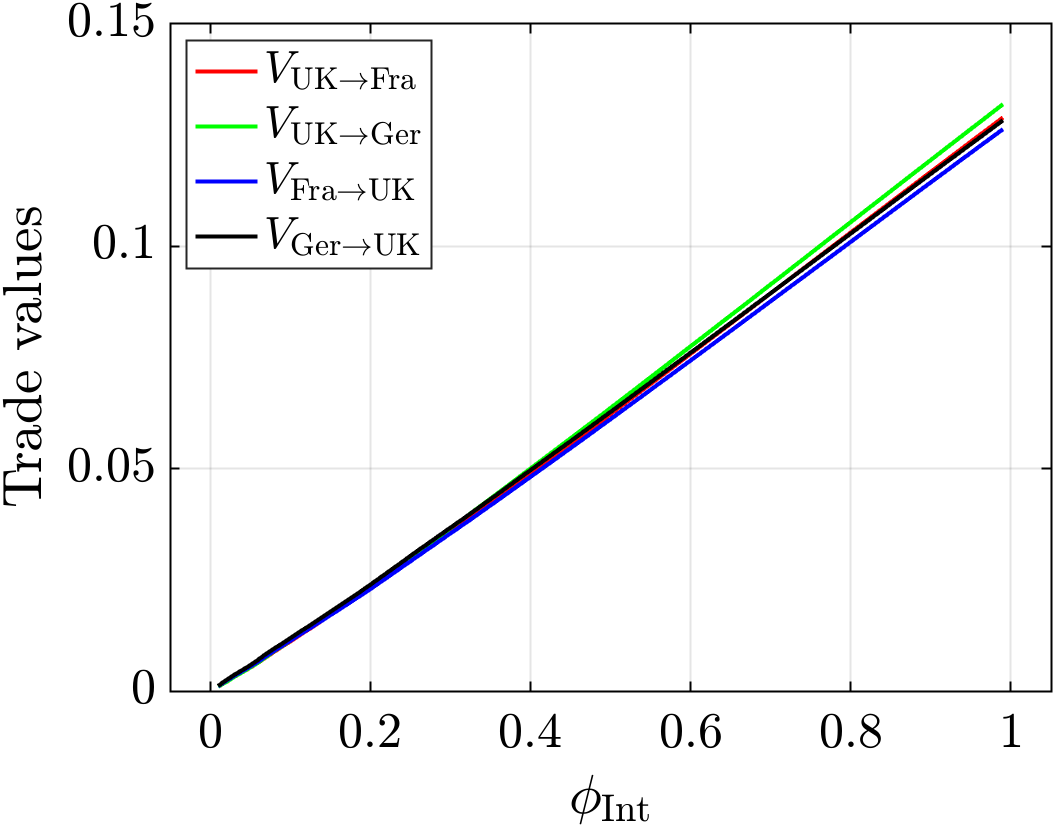} &
   \includegraphics[scale=0.37]{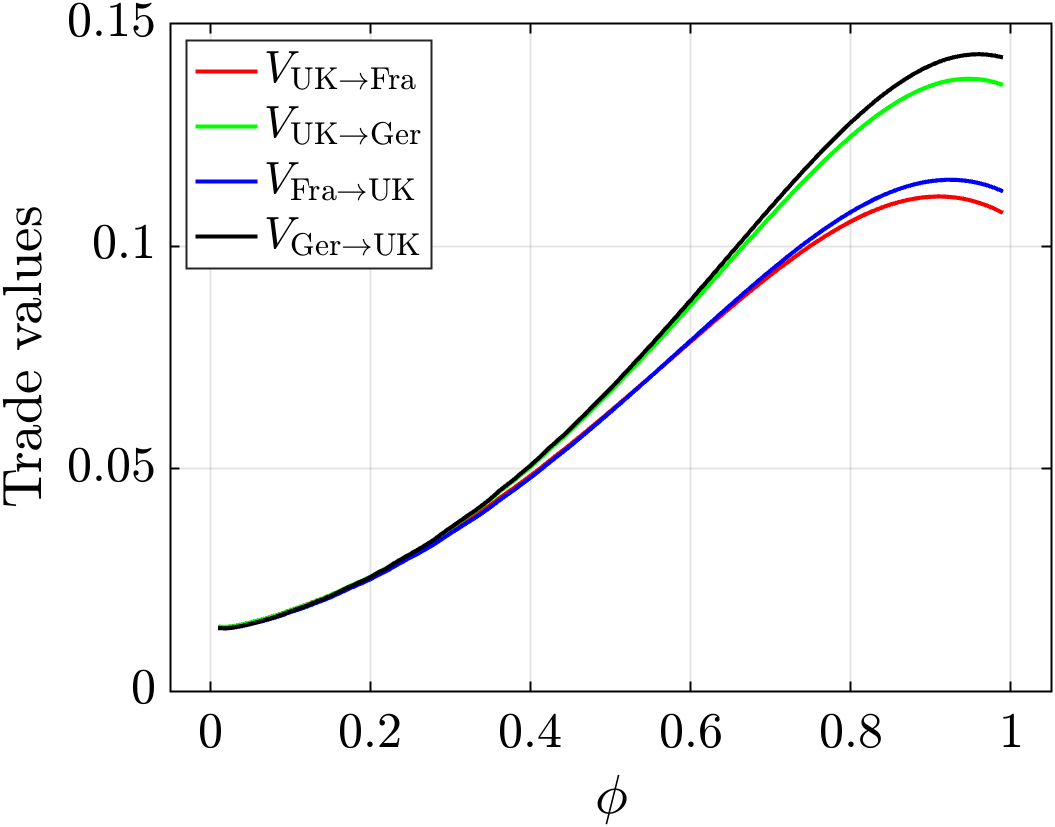} 
    \cr
  (a) Plot  against $\phi_{\rm Int}$ for $\phi=0.3$ &
  (b) Plot  against $\phi$ for $\phi_{\rm Int}=0.3$ 
\end{tabular}
   \caption{Dependence of trade values on trade-freeness parameters.
See \eqref{Tvolume-define} 
for the definition of these trade values.
}
   \label{TradeVolumeSce1}
\end{figure}

\OMIT{
\textcolor{red}{We may plot ??? the formula \eqref{decomposition_trade}
for the trade value:}
\begin{align*}
  \frac{\mathrm{d} V_{ij}}{\mathrm{d} \tau_k} = 
  m_i q_{ij} \lambda_j \frac{\mathrm{d} p_{ij}}{\mathrm{d} \tau_k} + 
  p_{ij} q_{ij} \lambda_j \frac{\mathrm{d} m_i}{\mathrm{d} \tau_k} + 
  m_i p_{ij} \lambda_j \frac{\mathrm{d} q_{ij}}{\mathrm{d} \tau_k} + 
  m_i p_{ij} q_{ij} \frac{\mathrm{d} \lambda_j}{\mathrm{d} \tau_k}.
\end{align*}
}


 \section{The UK's Post-Brexit Trade Strategy}\label{UK-Strategy}

The UK's departure from the EU in 2021
ended its economic integration with Europe.
We investigate the UK's trade policy to attract mobile workers 
in the post-Brexit period.

The UK seeks to expand its population share $\lambda_{\rm UK}$
by adjusting trade-freeness parameters
$\phi_{\rm UK}$ and $\phi_{\rightarrow{\rm UK}}$ 
 $(=\phi_{{\rm EU}\rightarrow{\rm UK}})$
(cf.~\eqref{phiIntAsym}).
 The EU, an alliance of France and Germany,
maintains a single market characterized by $\phi_{\rm EU}=0.3$
(cf.~\eqref{SeamlessEU}),
and chooses
between two types of tariffs on imports:
\begin{align*} &
\left\{
\begin{array}{ll}
\mbox{reciprocal tariff}: \quad &
\phi_{\rightarrow{\rm EU}}=
\phi_{\rightarrow{\rm UK}}, \cr
\mbox{asymmetric tariff}: \quad &
\phi_{\rightarrow{\rm EU}}=0.3.
\end{array}\right.
\end{align*}

\subsection{The UK's Trade Strategy at the Onset of Brexit}\label{AnalysisUK-TradePosition}

We analyze the UK’s instantaneous trade strategy 
at the onset of Brexit,
which we model as the origin point O
with $\lambda_{\rm UK}=0.309$ 
(cf.~Fig.~\ref{PathTracing gamma=0.3UK}).
In this analysis, we refer to the alliance-level population-gradient matrix
(${\rm L}={\rm A}$ in \eqref{GraDefine}),
 evaluated at this point, as
\begin{align*}
  T^{\rm A}=
   \begin{pmatrix}
   t_{{\rm UK},\phi_{\rm UK}} &
   t_{{\rm UK},\phi_{\rightarrow {\rm UK}}} \cr
   t_{{\rm EU},\phi_{\rm UK}} &
   t_{{\rm EU},\phi_{\rightarrow {\rm UK}}} \cr
 \end{pmatrix}=
 & 
 \left\{\begin{array}{ll}
 \begin{scriptsize}
 \begin{pmatrix}
+1.437 & +0.055 \cr
-1.437 & -0.055 \cr
 \end{pmatrix} 
 \end{scriptsize}
 & \mbox{under reciprocal tariff}, \cr 
 \noalign{\vskip 1ex}
 \begin{scriptsize}
 \begin{pmatrix}
+1.437 & ~+0.167 \cr
-1.437 &  ~-0.167 \cr
 \end{pmatrix} 
 \end{scriptsize}
~~ & \mbox{under asymmetric tariff}. 
 \end{array}\right.
 \end{align*} 
The gradient
$t_{{\rm UK},\phi_{\rm UK}}=1.437$ 
is much larger than the gradients 
$t_{{\rm UK},\phi_{\rightarrow {\rm UK}}}=0.055,~0.167$;
accordingly, $\phi_{\rm UK}$
is expected to exert greater influence on 
$\lambda_{\rm UK}$ than $\phi_{\rightarrow {\rm UK}}$.
Under any tariff type,
these gradients are positive;
accordingly, the UK is in a strong position 
with respect to increases both in
 $\phi_{\rm UK}$ and $\phi_{\rightarrow {\rm UK}}$.
The EU is in a weak position, 
as the strong and weak positions are reciprocal between the UK and the EU
 (cf.~Section~\ref{ReductionSingleCountry}).

The UK’s recommended instantaneous trade strategy just after Brexit 
is not contingent on the EU’s tariff choice and it is as follows.
The UK should invest in its infrastructure to increase 
domestic trade freeness and pursue trade liberalization 
to increase import trade freeness.
This strategy aligns with the argument by \citet{Brakman.etal.2018}, 
based on gravity-equation results:
 ^^ ^^ Paradoxically, only a trade agreement with the EU can compensate
 for Brexit's trade losses.”

\subsection{Analysis of Global Trade Strategy I: Contour Map}\label{UK'sContourMap}

We study the UK's global trade strategy for increasing
the UK's population share $\lambda_{\rm UK}$
with reference to the contour maps 
of $\lambda_{\rm UK}$ in Fig.~\ref{aUK-contour}. 
These maps are drawn in the whole parameter space of 
$(\phi_{\rightarrow{\rm UK}},\phi_{\rm UK})
\in (0,1)\times (0,1)$ and, accordingly,
provide global information on the effects of the parameters.

\begin{figure}
   \centering\small
\begin{tabular}{@{\hspace{-2.5mm}}cc}
 \includegraphics[scale=0.35]{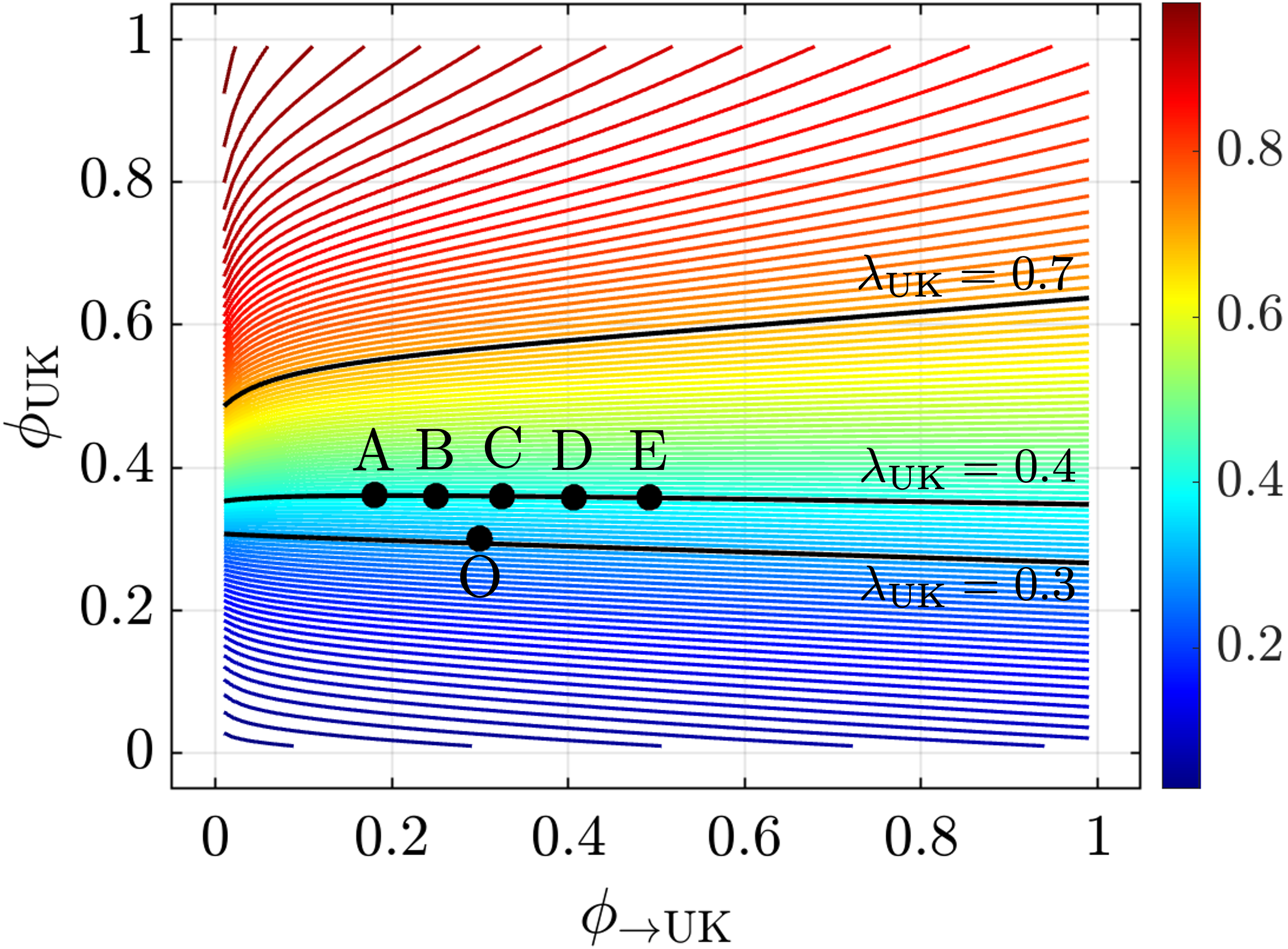} &
 \includegraphics[scale=0.35]{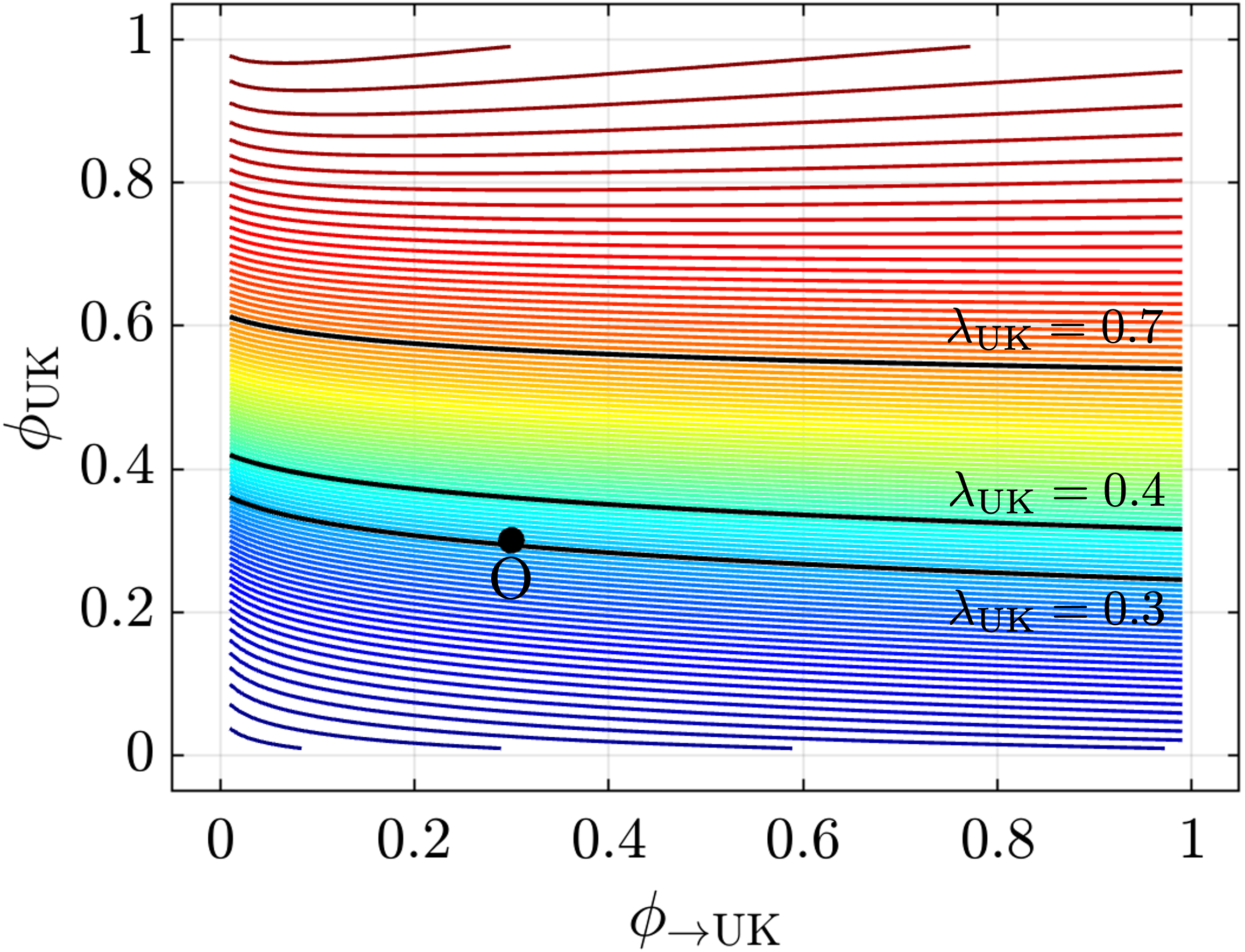} \\
(a) Reciprocal tariff  ($\phi_{\rightarrow{\rm EU}}=\phi_{\rightarrow {\rm UK}}$) &
(b) Asymmetric tariff ($\phi_{\rightarrow{\rm EU}}=0.3$) \\
  \noalign{\vskip 1.5ex}
 \multicolumn{2}{c}{\includegraphics[scale=0.3]{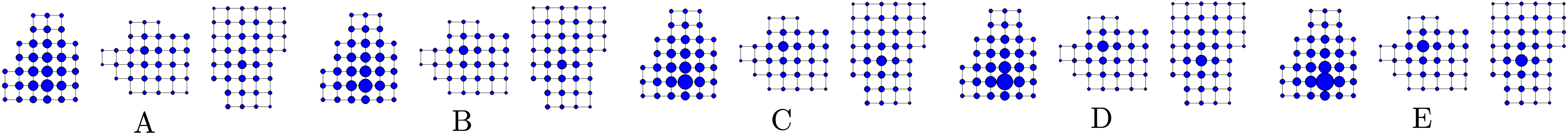}} \\
  \multicolumn{2}{c}{(c) Population distributions 
  with $\lambda_{\rm UK}=0.4$ (reciprocal tariff)} \\
  \noalign{\vskip .5ex}
  \end{tabular}
 \caption{
  Contour maps of $\lambda_{\rm UK}$ in the space of 
  $(\phi_{\rightarrow{\rm UK}},\phi_{\rm UK})$
  plotted under both tariff types.
  \OMIT{
We conducted an inverse analysis
(cf.~Section~\ref{FurtherReduction})
to identify 
  $(\phi_{\rightarrow{\rm UK}},\phi_{\rm UK})$
that yields the target country-level population distribution
$\bm{\lambda}^*=(\lambda_{\rm UK}^*,\lambda_{\rm Fra}^*,
\lambda_{\rm Ger}^*)$.
By targeting the points
$\bm{\lambda}^* = (0.40, 0.26+\eta, 0.34-\eta)$ for
several values of $\eta$,
we obtained points A--E in (a).
Similarly, the contours in (a) and (b) were obtained.
This corresponds to the case of $p=n^{\rm L}-1$
with $p=2$ and $n^{\rm L}=3$ in Section~\ref{FurtherReduction},
and an iteration was conducted to upgrade the accuracy of the inverse analysis
(cf.~\ref{InverseDetails}).
}
   }
    \label{aUK-contour}
  \end{figure}

First, we note that
$\lambda_{\rm UK}$ increases monotonically
with $\phi_{\rm UK}$ 
for any given $\phi_{\rightarrow{\rm UK}}$.
Thus, increasing $\phi_{\rm UK}$, 
through the domestic infrastructure promotion, 
always benefits the UK. 

Next, we investigate the effect of
the international trade freeness
$\phi_{\rightarrow{\rm UK}}$,
referring to the slope of the contour lines as follows:
 \begin{quote}
 When a contour line has a negative slope, 
  $\lambda_{\rm UK}$ increases with 
 $\phi_{\rightarrow{\rm UK}}$ 
for a constant $\phi_{\rm UK}$
and decreases when $\phi_{\rightarrow{\rm UK}}$ decreases.
 When it has a positive slope, this tendency is reversed.
 \end{quote}
When $\lambda_{\rm UK}=0.3$, the contour line has a negative slope under each type of tariff
(cf.~Fig.~\ref{aUK-contour}); accordingly,  trade liberalization favors the UK. 
If the infrastructure promotion of the UK
reaches a certain level, 
the slope becomes positive 
and trade liberalization becomes unfavorable
(e.g., $\lambda_{\rm UK}=0.7$ under the reciprocal tariff
in Fig.~\ref{aUK-contour}(a)).

\subsection{Analysis of Global Trade Strategy II: Trade Positions}\label{UK'sGlobalTrade}

We search for the general mechanism behind the reversed trend
observed above.
Using the gradient $t_{\alpha,\phi_{\rightarrow \alpha}}$
(slope) of the $\phi_{\rightarrow \alpha}$ versus $\lambda_{\alpha}$ curve,
we introduce Definition~\ref{TradePositionDefinition}.
\begin{definition}\label{TradePositionDefinition}
The trade position of country $\alpha$ is  
\begin{align}\label{MaturePremature} &
\left\{
\begin{array}{ll}
\mbox{premature} \quad & 
\mbox{if~}
t_{\alpha,\phi_{\rightarrow \alpha}}>0,
 \cr
\mbox{turning point (TP)} \quad & 
\mbox{if~}
t_{\alpha,\phi_{\rightarrow \alpha}}=0,
 \cr
\mbox{mature} \quad & 
\mbox{if~}
t_{\alpha,\phi_{\rightarrow \alpha}}<0.
\end{array}
\right. 
\end{align}
\end{definition}
\noindent
A country in a premature position 
is in a strong position
when $\phi_{\rightarrow \alpha}$ increases
and in a weak position when it decreases. 
Trade liberalization benefits a country in a premature position,
whereas protectionism benefits a country in a mature position.  

We propose Conjecture~\ref{MaturePrematureConj} on 
the transition of a country’s trade position
based on the discussion below.

\begin{conjecture}\label{MaturePrematureConj}
A premature position transitions into a mature one
when the country's domestic trade-freeness
parameter reaches a certain level.
\hfill $\Box$
\end{conjecture}

The UK's trade position depends
on the slope of the contour line and, in turn, 
on the trade-freeness parameters and the tariff type
(cf.~Fig.~\ref{aUK-contour}).
Figure~\ref{UKAnalysisParameter} shows
parameter zones of the UK's positions 
in the space of $(\phi_{\rightarrow{\rm UK}},\phi_{\rm UK})$.
For example, the origin point O at 
$(\phi_{\rightarrow {\rm UK}},\phi_{\rm UK})=(0.3,0.3)$
belongs to a premature trade position (colored red)
under both tariff types.
As $\phi_{\rm UK}$ increases from 0 along the vertical dashed line
at $\phi_{\rightarrow{\rm UK}}=0.3$,
the premature position transitions to a mature position colored blue 
at the point TP (turning point) at $\phi_{\rm UK}\approx 0.38$
under the reciprocal tariff (Panel (a)).
Under the asymmetric tariff (Panel (b)), TP resides at $\phi_{\rm UK}
\approx0.8$,
and the premature zone widens,
thereby favoring the UK under trade liberalization.

To sum up, the UK's suggested global trade strategy 
is to increase its national trade freeness.
The choice of either trade liberalization or protectionism
is contingent on
the EU's tariff policy for 
$0.38<\phi_{\rm UK} <0.80$
and is summarized 
in Table~\ref{UK-positions}
for several  ranges of $\phi_{\rm UK}$.

As shown in Fig.~\ref{UKAnalysisVolume}, 
exports from the UK to the EU increase with trade liberalization
(an increase in $\phi_{\rightarrow{\rm UK}}$),
especially under the reciprocal tariff.
Conversely, as $\phi_{\rightarrow{\rm UK}}$ decreases from its original value of 0.3---such as during Brexit---the UK exports decline sharply.
This decline is consistent with \citet{Freeman.etal.2025}:
 ^^ ^^ Our estimates imply that,
in the short term, leaving the EU 
reduced worldwide UK exports by 6.4\% 
and worldwide imports by 3.1\%.”

 \begin{figure}
   \centering\small
\begin{tabular}{cc}
 \includegraphics[scale=0.4]{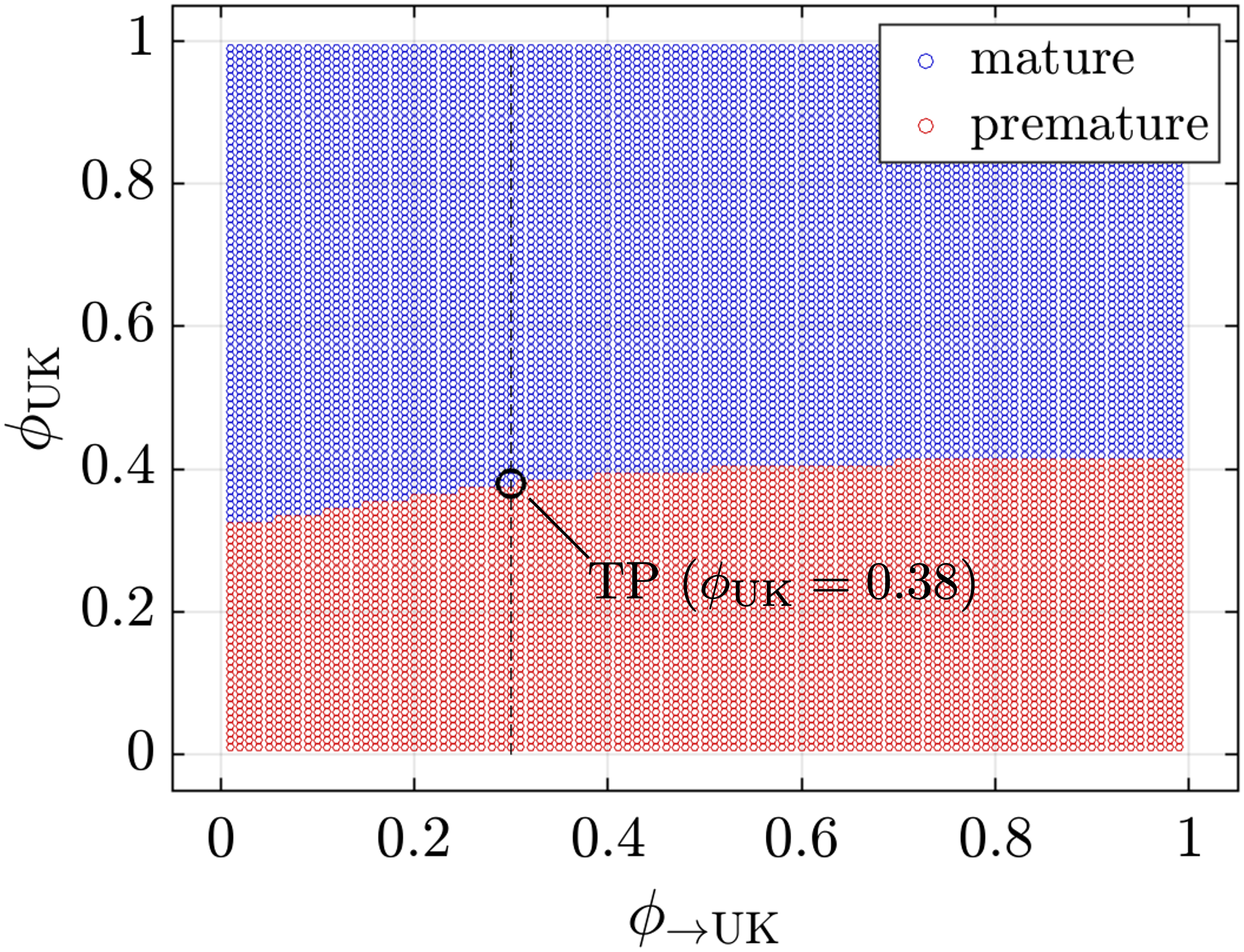} &
 \includegraphics[scale=0.4]{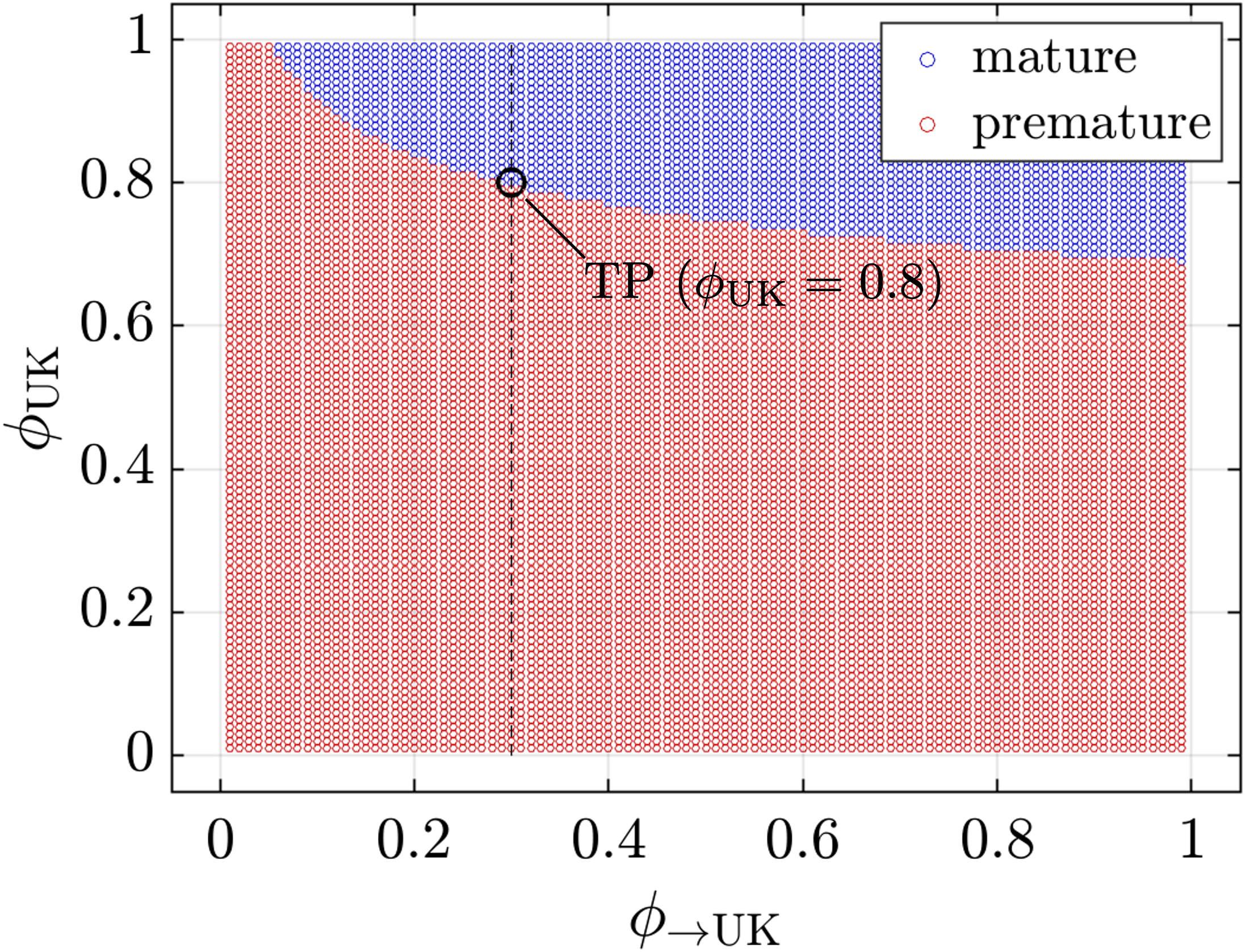} \\ 
(a) Reciprocal tariff  ($\phi_{\rightarrow{\rm EU}}=\phi_{\rightarrow {\rm UK}}$) &
(b) Asymmetric tariff ($\phi_{\rightarrow{\rm EU}}=0.3$) \\
\end{tabular}
 \caption{Parameter zones of trade positions for the UK
in the space of $(\phi_{\rightarrow {\rm UK}},\phi_{\rm UK})$.
The zone of a premature position is colored red
and the zone of a mature one is colored blue.
 } 
   \label{UKAnalysisParameter}
\vspace{-2mm}
 \end{figure}

\begin{table}
\centering
\caption{Classification of the UK's 
 trade position
and suggested trade policy
($\phi_{\rightarrow{\rm UK}}\nearrow$ and 
$\phi_{\rightarrow{\rm UK}}\searrow$:
an increase and a decrease of
$\phi_{\rightarrow{\rm UK}}$)
}\label{UK-positions}
\vspace{-2mm}
\begin{tabular}{cl|lc}
Range of $\phi_{\rm UK}$       &  Tariff type &UK's position &  Suggested trade policy
 \\ \hline
$0.30<\phi_{\rm UK} <0.38$  & Both types &  Premature & Liberalization 
($\phi_{\rightarrow{\rm UK}}\nearrow$) \cr
$0.38<\phi_{\rm UK} <0.80$  & Reciprocal &  Mature & Protectionism 
($\phi_{\rightarrow{\rm UK}}\searrow$) \cr
                                                   & Asymmetric& Premature & Liberalization
($\phi_{\rightarrow{\rm UK}}\nearrow$) \cr
$0.80<\phi_{\rm UK} <1.00$   & Both types & Mature &  Protectionism
($\phi_{\rightarrow{\rm UK}}\searrow$) \cr
\end{tabular}
\end{table}

\begin{figure}
   \centering\small
\begin{tabular}{cc}
\includegraphics[scale=0.4]{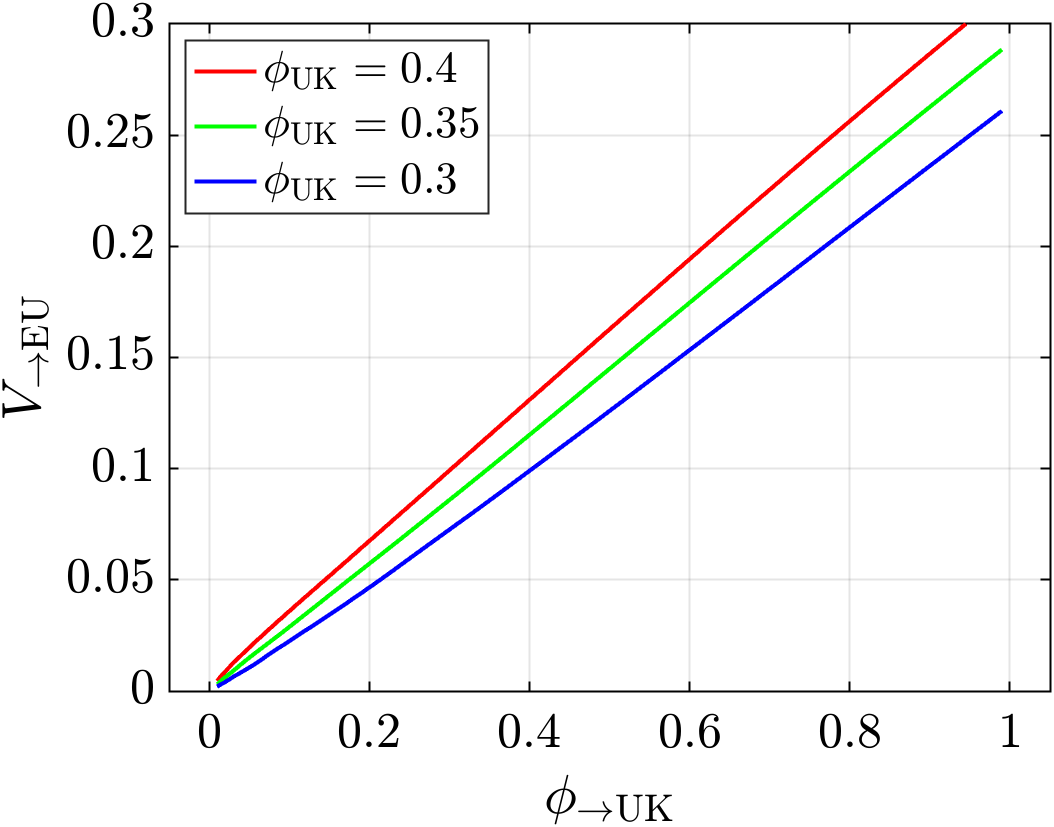} & \includegraphics[scale=0.4]{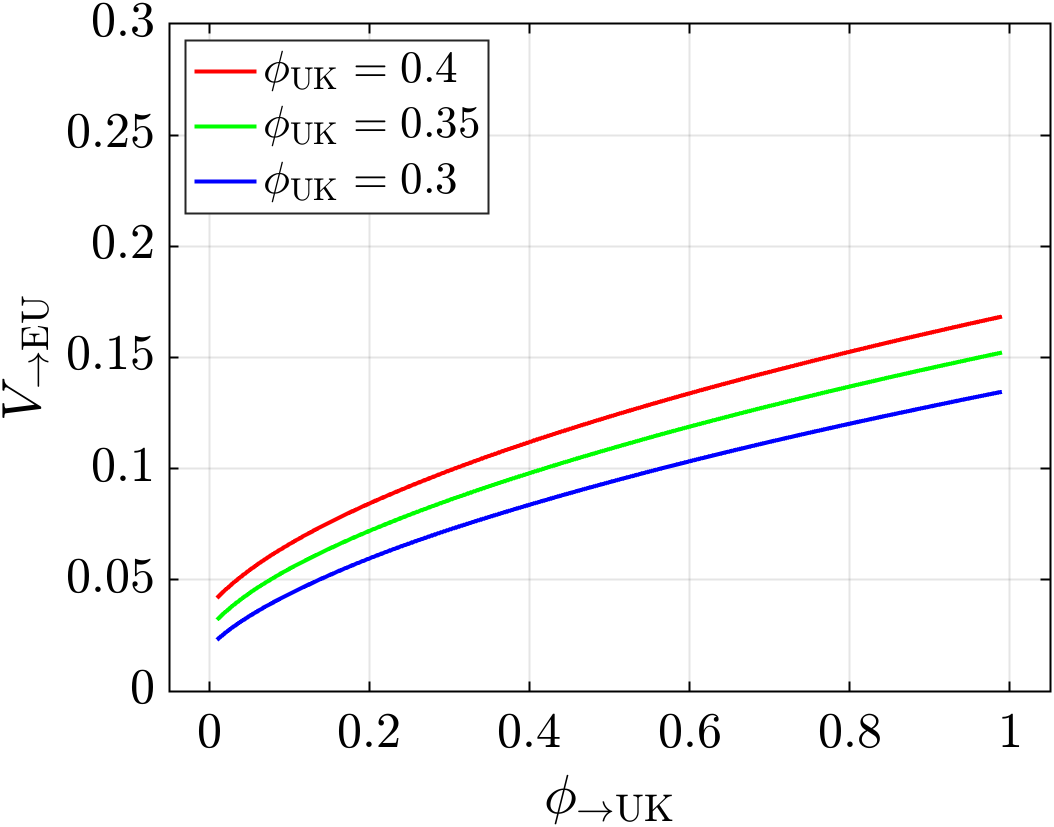} \\ 
(a) Reciprocal tariff  ($\phi_{\rightarrow{\rm EU}}=\phi_{\rightarrow {\rm UK}}$) &
(b) Asymmetric tariff ($\phi_{\rightarrow{\rm EU}}=0.3$)
\end{tabular}

 \caption{The influence of trade liberalization
 on trade value $T_{\rightarrow{\rm EU}}$ of the UK exports to the EU.
    See \eqref{Tvolume-define} in \ref{TradeValueEachCountry}
for the definition of this trade value.
   The UK's imports displayed the same trend as
  the exports and are therefore omitted here.
 } 
   \label{UKAnalysisVolume}
\vspace{-5mm}
  \end{figure}


 \section{The EU's Post-Brexit Trade Strategy}\label{EU-Strategy}

We analyze the EU's trade strategy
to increase its population share
$\lambda_{\rm EU}$ by adjusting
its domestic and trade freeness parameters $\phi_{\rm EU}$ 
and $\phi_{\rightarrow{\rm EU}}$.
The UK maintains its domestic trade freeness
at $\phi_{\rm UK}=0.3$ 
and chooses between two types of tariffs on imports from the EU:
\begin{align} \label{UK'sTrade} &
\left\{
\begin{array}{ll}
\mbox{reciprocal tariff}: \quad &
\phi_{\rightarrow{\rm UK}}=
\phi_{\rightarrow{\rm EU}}, \\ 
\mbox{asymmetric tariff}: &
\phi_{\rightarrow{\rm UK}}=0.3.
\end{array}
\right.
\end{align}

\subsection{The EU's Instantaneous Trade Strategy}

We analyze the EU's instantaneous trade strategy
at the origin point O
(cf.~Fig.~\ref{PathTracing gamma=0.3UK}),
which represents the state just before Brexit.
The population-gradient matrix at this point is given by
\begin{align*}
  T^{\rm A}=
   \begin{pmatrix}
   t_{{\rm EU},\phi_{\rm EU}} &
   t_{{\rm EU},\phi_{\rightarrow {\rm EU}}} \cr
   t_{{\rm UK},\phi_{\rm EU}} &
   t_{{\rm UK},\phi_{\rightarrow {\rm EU}}} \cr
 \end{pmatrix}=
 & \left\{\begin{array}{ll}
 \begin{scriptsize}
 \begin{pmatrix}
+1.458&~  -0.055 \cr
-1.458 & ~ +0.055 \cr
 \end{pmatrix} 
 \end{scriptsize}
 & \mbox{under reciprocal tariff}, \cr
 \noalign{\vskip 1ex}
 \begin{scriptsize}
 \begin{pmatrix}
+1.458 &~ +0.113 \cr
-1.458 &~ -0.113 \cr
 \end{pmatrix} 
 \end{scriptsize}
 \ \ & \mbox{under asymmetric tariff}.
 \end{array}\right.
 \end{align*} 
First, we note that 
$t_{{\rm EU},\phi_{\rm EU}}=1.458$ is positive
for any tariff type.
Accordingly, the EU holds a strong position as $\phi_{\rm EU}$ increases 
and benefits from infrastructure development 
that enhances $\phi_{\rm EU}$.
In contrast, the UK has a negative gradient $t_{{\rm UK},\phi_{\rm EU}}=-1.458$
and therefore in a weak position.
 
Next, we analyze the influence of 
the import trade freeness 
$\phi_{\rightarrow {\rm EU}}$.
Under the reciprocal tariff,
the EU has a negative gradient  
$t_{{\rm EU},\phi_{\rightarrow {\rm EU}}}=-0.055$
and is in a mature position
(cf.~\eqref{MaturePremature}).
Trade liberalization (an increase in $\phi_{\rightarrow {\rm EU}}$)
reduces $\lambda_{\rm EU}$ and is unfavorable for the EU.
Under the asymmetric tariff, we observe a reversed trend:
Since $t_{{\rm EU},\phi_{\rightarrow {\rm EU}}}=0.113>0$,
the EU is in a premature position and 
 $\lambda_{\rm EU}$ increases with $\phi_{\rightarrow {\rm EU}}$.
Thus, 
the EU's success or failure of
tariff changes depends
on the UK's tariff choices. 
Shifts in the international trade environment 
act as a double-edged sword, 
that is, the strategic gamble summarized in Table~\ref{WinnerEU}.
The recommended trade strategy for the EU is to
maintain current tariff levels---avoiding risky moves---
while increasing national trade freeness
through infrastructure development.

\begin{table}
\caption{Trade strategies of the EU and the UK and the resulting 
population gainer}
\label{WinnerEU}
\centering
\begin{small}
\begin{tabular}{cc|c} 
\multicolumn{2}{c|}{Country's Tariff Policy} & 
Winner \\ 
EU & UK & (Population Gainer) \\ \hline 
Reduce Tariff & Hold Tariff & EU \\
 & Reduce Tariff & UK \\ \hline
Raise Tariff & Raise Tariff & EU \\
& Hold Tariff & UK \\  \hline
\end{tabular}
\end{small}
\end{table}

\subsection{The EU's Global Trade Strategy}

To design the EU's global trade strategy under both tariff types in \eqref{UK'sTrade},
we refer to the contour maps of the EU's population 
share $\lambda_{\rm EU}$
in Fig.~\ref{aEU-contour}. 
First, 
we note that $\lambda_{\rm EU}$ increases with $\phi_{\rm EU}$
for any $\phi_{\rightarrow{\rm EU}}$ and any tariff type.
Thus, increasing $\phi_{\rm EU}$ benefits the EU.

Next, we investigate the influence of 
$\phi_{\rightarrow{\rm EU}}$.
Whether $\lambda_{\rm EU}$ increases or decreases,
as $\phi_{\rightarrow{\rm EU}}$ changes, 
depends on whether the contour in Fig.~\ref{aEU-contour}
has a negative or positive slope,
respectively (cf.~Section~\ref{UK'sContourMap}).
The mechanism of this dependence can be 
understood from the parameter zones of the EU's trade positions 
in Fig.~\ref{EUAnalysisParameter}
that vary with the values of the parameters 
$(\phi_{\rightarrow{\rm EU}},\phi_{\rm EU})$
and with the tariff type.
As $\phi_{\rm EU}$ increases from 0, 
a premature position (negative slope) colored red transitions 
to a mature position (positive slope) colored blue
for any $\phi_{\rightarrow{\rm EU}}$, 
as in the UK's case
 (cf.~Fig.~\ref{UKAnalysisParameter}).

\begin{figure}
   \centering\small
\begin{tabular}{@{\hspace{-2.5mm}}cc}
 \includegraphics[scale=0.38]{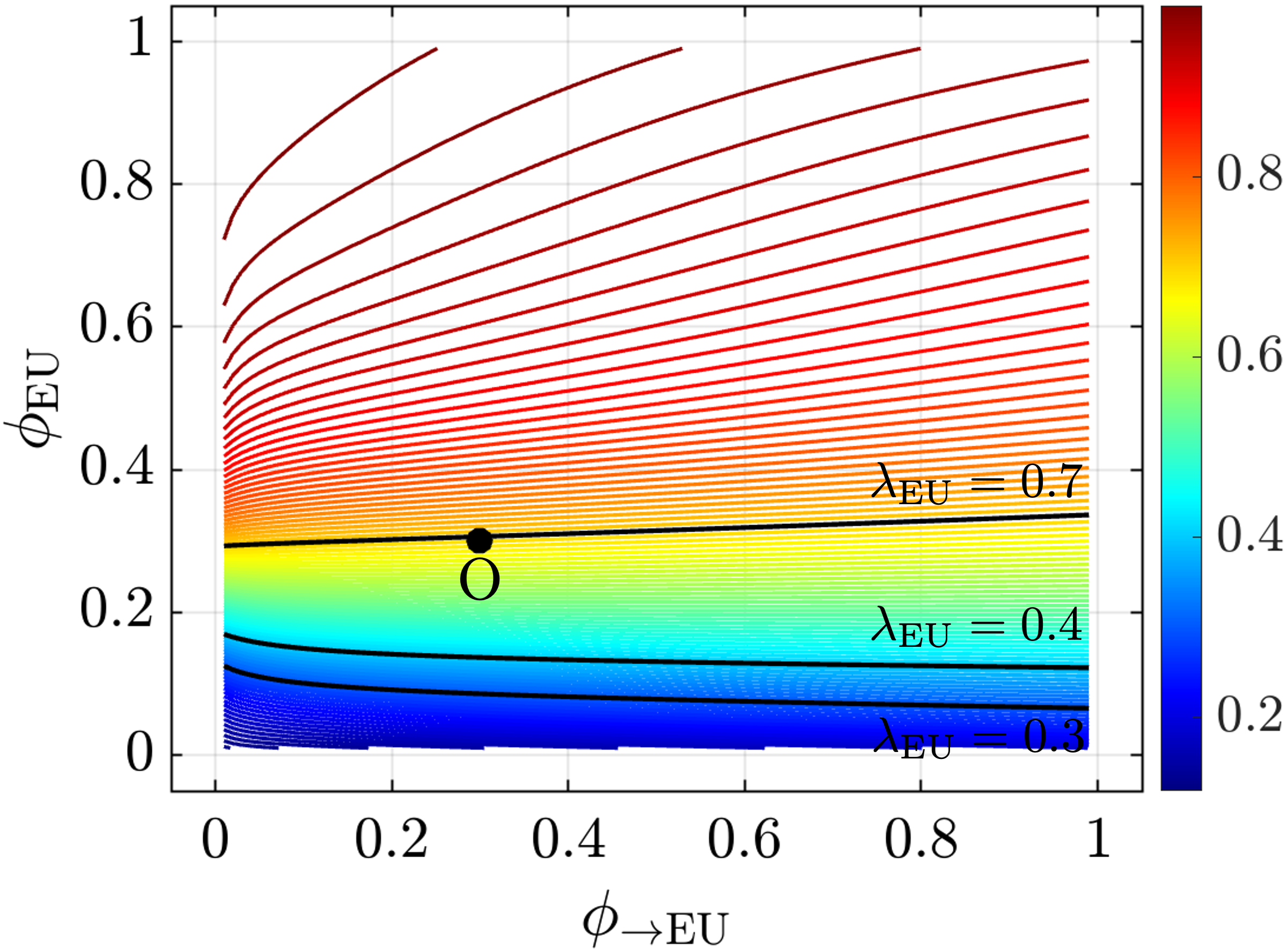} &
 \includegraphics[scale=0.38]{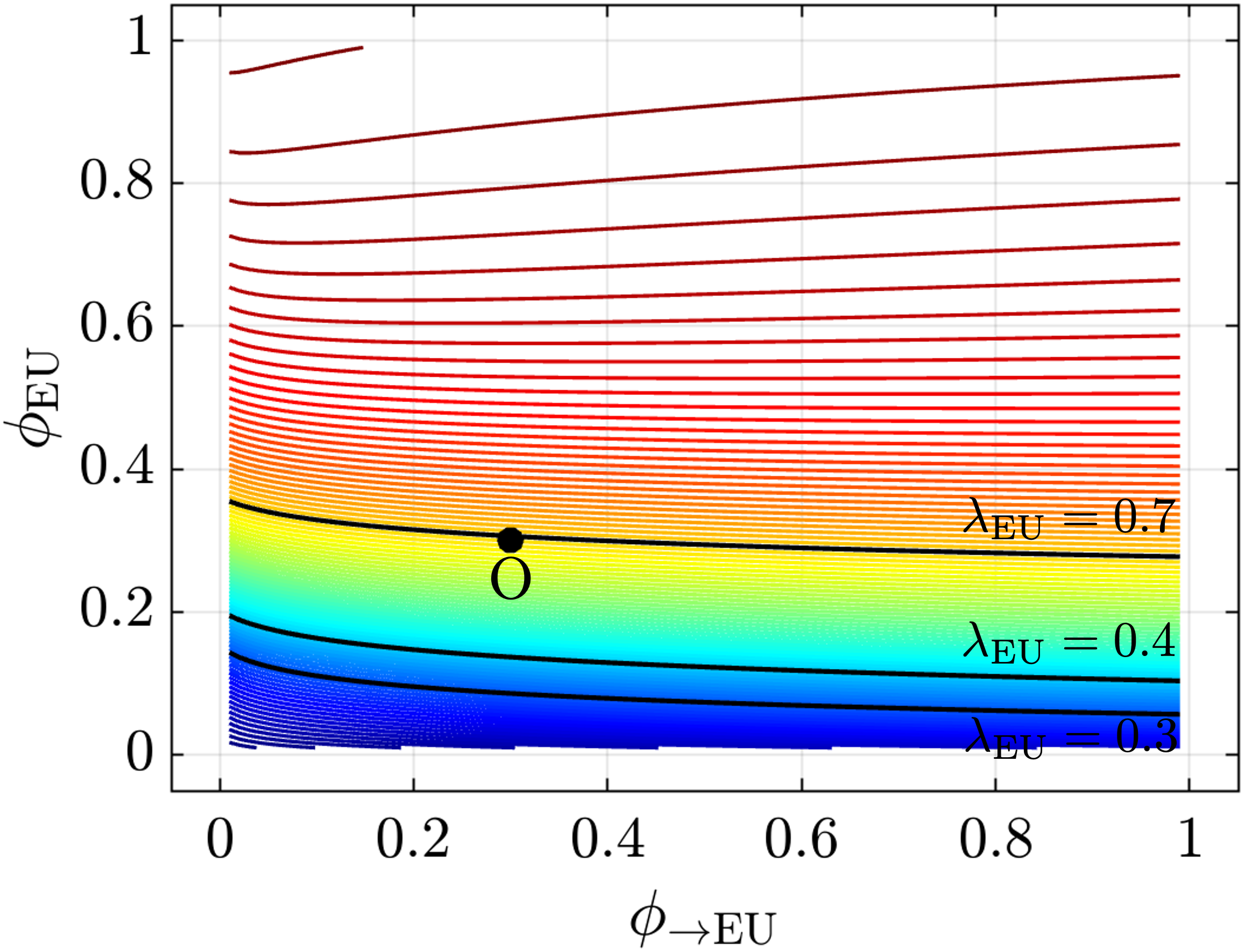} \\
(a) Reciprocal tariff  ($\phi_{\rightarrow{\rm UK}}=\phi_{\rightarrow{\rm EU}}$) &
(b) Asymmetric tariff ($\phi_{\rightarrow{\rm UK}}=0.3$)
  \end{tabular}
 \caption{
  Contour maps of $\lambda_{\rm EU}$ in the space of 
  $(\phi_{\rightarrow{\rm EU}},\phi_{\rm EU})$
  plotted under both tariff types.
   }
    \label{aEU-contour}

    \vspace{5mm}

   \centering\small
\begin{tabular}{cc}
 \includegraphics[scale=0.38]{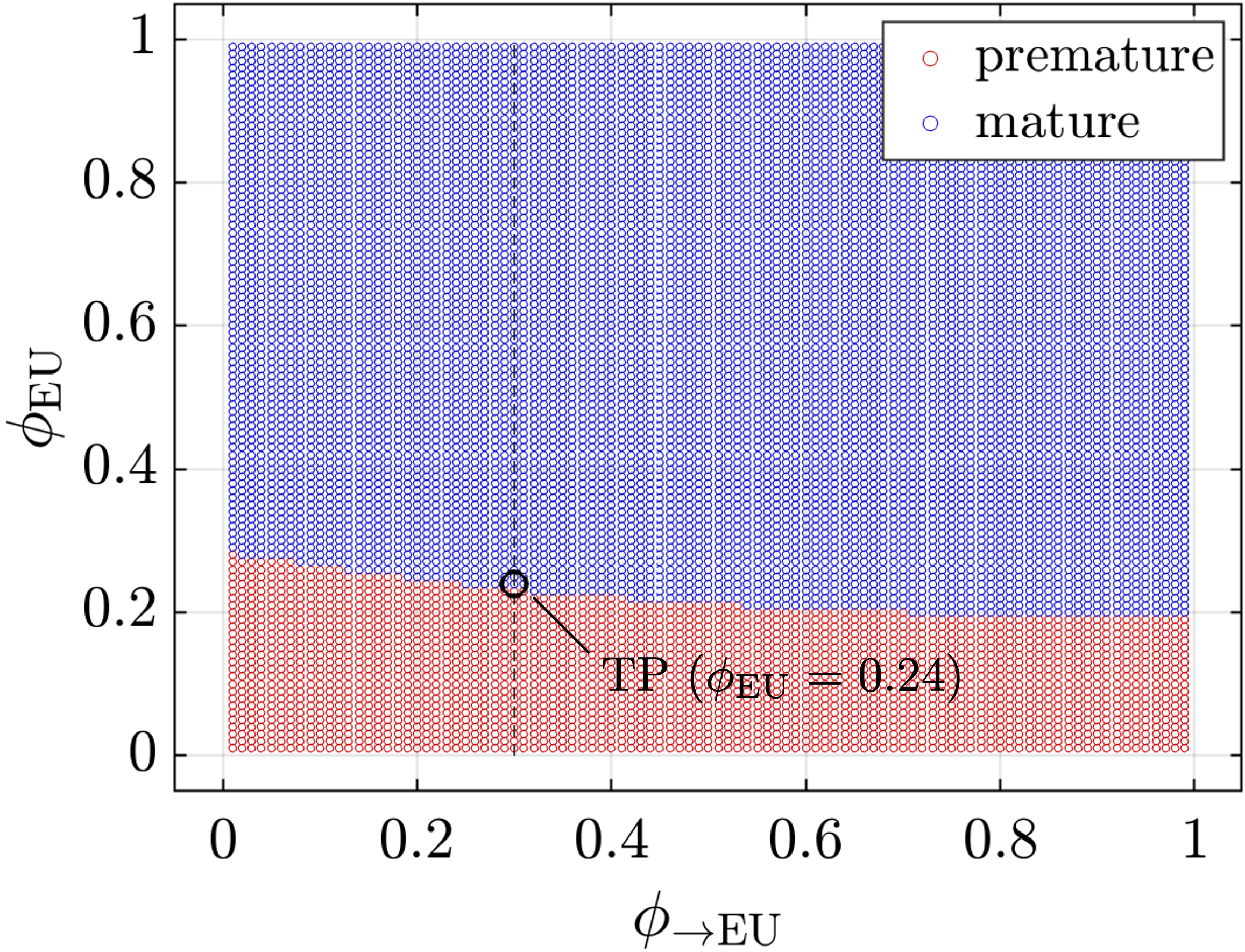} &
 \includegraphics[scale=0.38]{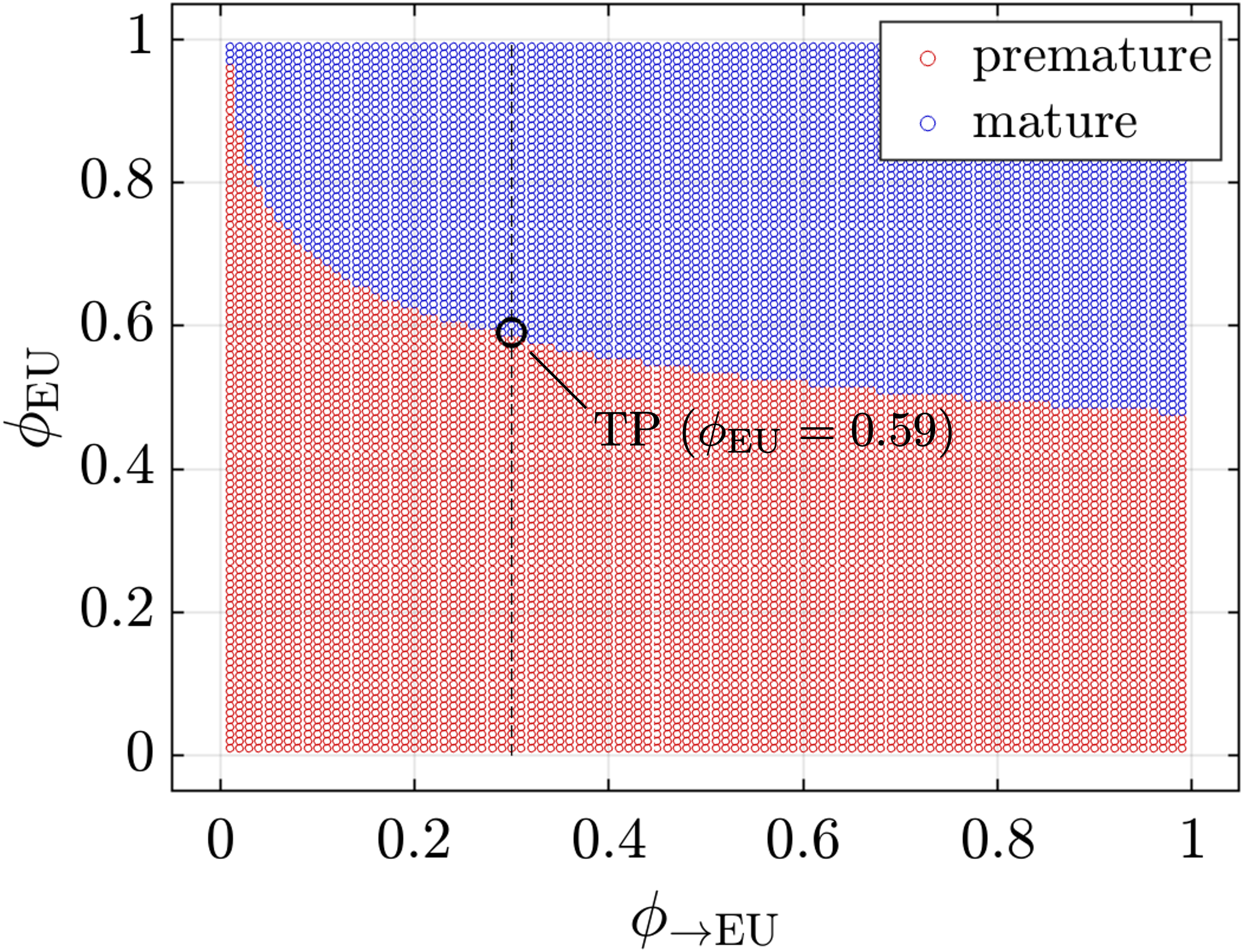} \\ 
(a) Reciprocal tariff  ($\phi_{\rightarrow{\rm UK}}=\phi_{\rightarrow{\rm EU}}$) &
(b) Asymmetric tariff ($\phi_{\rightarrow{\rm UK}}=0.3$)
\end{tabular}
 \caption{ Parameter zones of trade positions for the EU
in the space of $(\phi_{\rightarrow {\rm EU}},\phi_{\rm EU})$.
The zone of a premature position is colored red
and the zone of a mature one is colored blue.
 } 
   \label{EUAnalysisParameter}
  \end{figure}

The origin point O at 
$(\phi_{\rightarrow {\rm EU}},\phi_{\rm EU})=(0.3,0.3)$
is in a mature position under the reciprocal tariff
but in a premature one under the asymmetric tariff.
In contrast,
the origin point of the UK is in a premature position under both tariff types
(cf.~Fig.~\ref{UKAnalysisParameter}).
The EU and the UK have different 
locations of the turning points TP, although they share the same mechanism governing the transition of their trade positions.
The turning point is located below the origin point $(\phi_{\rm EU}=0.24)$
for the present case
but above the origin point
 ($\phi_{\rm UK}=0.38$ in Fig.~\ref{UKAnalysisParameter})
in the case of the UK's strategy.

To sum up, the EU's suggested global trade strategy 
is to increase its national trade freeness.
The EU's choice of either trade liberalization or protectionism
is contingent on the UK's tariff policy and is summarized 
in Table~\ref{EU-positions}
for several  ranges of $\phi_{\rm EU}$.

\bigskip
\begin{table}[h!]
\centering
\caption{Classification of the EU's
 trade position
and suggested trade policy
($\phi_{\rightarrow{\rm EU}}\nearrow$ and 
$\phi_{\rightarrow{\rm EU}}\searrow$:
an increase and a decrease of
$\phi_{\rightarrow{\rm EU}}$, respectively)
}\label{EU-positions}
\begin{tabular}{cc|cc}
Range of $\phi_{\rm EU}$       &  Tariff type &EU's position &  Suggested trade policy
 \\ \hline
$0.30<\phi_{\rm EU} <0.59$  & Reciprocal &  Mature & Protectionism 
($\phi_{\rightarrow{\rm EU}}\searrow$) \cr
                                                   & Asymmetric& Premature & Liberalization
($\phi_{\rightarrow{\rm EU}}\nearrow$) \cr
$0.59<\phi_{\rm EU} <1.00$   & Both types & Mature &  Protectionism
($\phi_{\rightarrow{\rm EU}}\searrow$) \cr
\end{tabular}
\end{table}

\section{Concluding Remarks}\label{sec:global_concl}

This study presented a reduction 
method for analyzing international trade competition 
within a hierarchical spatial system in economic geography.
The original region-level governing equation
was reduced to a country-level equation 
and an alliance-level equation both with 
substantially fewer degrees of freedom,
thereby yielding insights into 
international trade competition.
We then developed analytical methods based on these reduced equations.
The population-gradient matrix quantifies the
influence of trade costs 
and helps design an instantaneous strategy 
for a given trade environment, such as Brexit.
The inverse analysis generates contour maps of the target country’s population share, providing a global view of the influence of these parameters.

We applied the proposed method to analyze trade competition among 
the UK, France, and Germany
during the Brexit period, incorporating internationally mobile workers.
The analysis reveals several mechanisms driving international migration.
The population-gradient matrix helps analyze how variations in transport costs affect country-level populations.
We demonstrated the crucial role of domestic infrastructure development
and clarified how 
tariff policies can attract mobile workers.

We considered a pre-Brexit EU single-market scenario in which the three countries jointly designed their national and international trade environments.
Ironically, trade liberalization benefits the UK the most,
even though it was the country that initiated 
economic disintegration through Brexit.

In the post-Brexit trade environment,
the UK can gain mobile population primarily through
domestic infrastructure development.
Trade liberalization benefits the UK irrespective of the EU’s tariff policy.
The recommended strategy for the UK is to enhance 
both domestic and import trade freeness.
Although trade liberalization favors the UK in the original state, 
the UK’s trade position would undergo a phase shift 
after a significant increase in the domestic trade freeness,
making protectionism favorable for the UK.
 
In the EU’s post-Brexit trade strategy,
investing in its domestic infrastructure
always favors the EU.
Altering the tariff level is a double-edged sword,
because the consequence depends critically on the UK’s response.
When international trade freeness is increased by reducing the tariff,
the EU loses population under the reciprocal tariff
but gains population 
if the UK retains its tariff level.

When a country raises a tariff against another,
how the affected country responds is critical.
There may be a political temptation to 
introduce a retaliatory tariff.
However, our results indicate that retaliatory tariffs are detrimental
to both the UK and the EU.
Although tariffs can be imposed quickly and may seem attractive,
this study emphasizes the importance of infrastructure
development.
Such development takes time---possibly many years, spanning the terms of
several presidents or prime ministers.

While this study employs the Helpman (1998) model
in the analysis, 
the proposed approach applies to a broad class of general-equilibrium environments.
For example, it can be used to study how tariff policies shape the spatial distribution of heterogeneous workers and capital across countries, as well as the cross-border relocation of firms.
It is  left for future research to apply this approach to other economic geography models
 (cf.~\citealp{Zeng.Zhao.2010},
\citealp{Fajgelbaum.Gaubert.2020}, and
\citealp{Janeba.Schulz.2024}).

We conclude this study with a remark on the US's trade tactics.
The US raised the tariffs on steel and aluminum imports
during the President Trump's first term
and ^^ ^^ metal production picked up,
but higher costs slowed other industries”
\citep{Sachdev.Rao.2025}.
The analysis presented in this study 
provides insight into the consequences 
of the high tariffs introduced by the US 
 in 2025.
The US can raise tariffs 
first and then invest in domestic infrastructure,
possibly using some of the tariff revenue.
Such tactics may strain other countries but
can also undermine the US economy.
These tactics are fragile, as their success
depends critically on how other countries respond.
A critical task for future research 
is to investigate further the mechanisms of trade competition involving the US.

\OMIT{
\section*{Acknowledgements}
We gratefully acknowledge the essential academic guidance of Prof. J.-F. Thisse. 
This work has received Grants-in-Aid from JSPS 21K04299/24K22968/\allowbreak24K16372/25K01339/25H00543
and Fusion Oriented Research for Disruptive Science and Technology (Grant No. JPMJFR215M).
}



\begin{singlespace}

\end{singlespace}






\appendix 

\section{Details of the Helpman Model}\label{HelpmanDetails}

\subsection{Basic Framework}

Preferences over housing and differentiated goods are identical for all workers, as shown in Section~\ref{HelpmanModel}.
Each place hosts a continuum of firms.
Each firm produces a single type of differentiated good using only labor as the factor of production. 
To produce $x_i(\varphi)$ units of the $\varphi$th differentiated good, 
each firm requires $f + c \, x_i(\varphi)$ units of labor.
Accordingly, the total production cost of a firm located in place $i$ is 
$w_i(f+c \, x_i(\varphi))$.
Each firm located in place $i$ maximizes its profit which is given by
\begin{align}
    \pi_i(\varphi) = \sum_{j \in N} p_{ij}(\varphi) q_{ij}(\varphi)\lambda_j 
    - w_i \left\{f+ c \, x_i(\varphi)\right\},
\end{align}
where $q_{ij}(\varphi)\lambda_j$ denotes total demand in place $j$ for the $\varphi$th differentiated good produced in place $i$.
The labor market is perfectly competitive, and all firms treat wages as given.

Under iceberg transportation costs 
(cf.~Section~\ref{IcebergTransportationCost}), only a fraction $1/\tau_{ij}$ of a unit shipped from place $i$ to place $j$ arrives.
Consequently, the total supply satisfies
 $x_i(\varphi)=\sum_{j \in N}\tau_{ij} \, q_{ij}(\varphi)\lambda_j$.
Hence, the first-order condition for profit maximization yields 
\begin{align}\label{price_variety}
    & p_{ij}(\varphi) = \frac{\sigma}{\sigma - 1} \, c \tau_{ij} w_i,
\end{align}
which is identical for all differentiated goods.
Likewise, $q_{ij}(\varphi)$ and $\pi_i(\varphi)$
are independent of $\varphi$.
We hereafter omit the argument of these variables.

In this model, market clearing holds for a given spatial 
distribution vector $\bm{\lambda}$ of mobile workers.  
Under the land market-clearing condition,
 the housing stock at each place satisfies
 $S = \lambda_i h_i$.
Substituting $h_i=\frac{(1 - \mu)Y_i }{r_i}$ in \eqref{qijhivi}
into $S = \lambda_i h_i$
gives the equilibrium housing price:
\begin{align}\label{housing_price_equilibrium}
  r_i = \frac{(1 - \mu) Y_i \lambda_i }{S}.
\end{align}
Combining this equation with 
the assumption of public ownership
yields the expenditure of a worker residing in place $i$: 
\begin{align}
    Y_i = w_i + (1 - \mu) \sum_{j \in N} \lambda_j Y_j.
\end{align}
Using this equation, we express equilibrium expenditure as a function of 
the wage vector $\bm{w}=(w_i)$ 
and the distribution vector $\bm{\lambda}=(\lambda_i)$ as
\begin{align}\label{income_wage}
    \bm{Y} = 
   \left[
        I - (1 - \mu) \bm{1} \, \bm{\lambda}^\top 
    \right]^{-1} \bm{w},
\end{align}
where 
$\bm{Y}=(Y_i)$,
$\bm{1} = ({1, \ldots, 1})^{\top}
 \in \mathbb{R}\sp{n}$, 
and $I$ denotes the $n \times n$ identity matrix.

The labor market-clearing and zero-profit conditions are given by 
\begin{align}
\label{LaborClearZeroProfit}
  &
    \lambda_i = m_i
    \left[
        \left(\sum_{j\in N} c \lambda_j q_{ij} \tau_{ij}\right)
         + f 
    \right], \quad
  \sum_{j \in N} \lambda_j q_{ij} (p_{ij} - c \tau_{ij} w_i ) - w_i f = 0.
\end{align}
Substituting the price from \eqref{price_variety} into
the second relation of \eqref{LaborClearZeroProfit} yields 
\begin{align}\label{zero_profit_deform} 
  \frac{1}{\sigma - 1}\sum_{j\in N} c \lambda_j q_{ij} \tau_{ij}  = f.
\end{align} 
Using the first relation in \eqref{LaborClearZeroProfit}, 
together with the above equation, gives
  $m_i = \frac{\lambda_i}{\sigma f}$.
Substituting this relation and \eqref{price_variety} into the price index 
\eqref{price_index} yields 
\begin{align}\label{price_index_equilibrium}
  P_i = 
  \left[
      \sum_{j\in N}
      \frac{\lambda_j}{f\sigma}
      \left(
          \frac{\sigma}{\sigma - 1}
          c \tau_{ji} w_j
      \right)^{1 - \sigma}
  \right]^{1 / (1 - \sigma)}.
\end{align} 
Substituting $q_{ij}$ given in \eqref{qijhivi}
into \eqref{zero_profit_deform}
and using  \eqref{price_index_equilibrium},
we obtain the following fixed-point problem for the wage:
\begin{align}\label{wage_equation}
    \mu
    \sum_{j\in N}
    \left[
   { (\lambda_j Y_j \tau_{ij}^{1 - \sigma} w_i^{1 - \sigma}) / }{ \left(\sum_{k\in N} \lambda_k \tau_{kj}^{1 - \sigma}
        w_k^{1 - \sigma} \right)
    }
    \right] 
    = w_i.
\end{align}
This equation implicitly determines the equilibrium wage $w_i$.
Substituting the price index \eqref{price_index_equilibrium} 
and the housing price \eqref{housing_price_equilibrium} into the indirect utility
function $v_i$ in \eqref{qijhivi} gives 
\begin{align}
    v_i =  \zeta \,
    \lambda_{i}^{\mu - 1} Y_i^\mu 
    \left(
      \sum_{j \in N} \lambda_j (\tau_{ji} w_j)^{1 - \sigma}
    \right)^{\mu / (\sigma - 1)},
\end{align}
where $\zeta$ is a constant that depends on exogenous variables:
\begin{align*}
\zeta=
  \left(
    \frac{1}{f \sigma}
  \right)^{-\mu / (1 - \sigma)}
  \left(\frac{\sigma c}{\sigma - 1}\right)^{-\mu}
  \left(
    \frac{1 - \mu}{S}
  \right)^{-(1 - \mu)}.
\end{align*}

\subsection{Parameter Sensitivity of Economic Variables}

In the Helpman model, we can derive explicit expressions for the derivatives of the functions with respect to the transportation cost $\tau_k$.

To simplify the presentation of matrix expressions in the sequel, 
we introduce the following short-hand expressions: 
\begin{align*}
  & \bm{x}^{\alpha} \equiv (x_i^\alpha)_{i},
  \quad
   (\bm{x}^\alpha \bm{y}^\beta)^* \equiv (x_i^\alpha y_i^\beta)_{i},
  \quad
  (\bm{x} A)^\dagger  \equiv ( x_i A_{ij} )_{i, j},
\end{align*}
and so on.
Here, 
$\bm{x}$ and $(\cdot)^*$ are column vectors, 
and $A$ and $(\cdot)^\dagger$ are matrices.

\subsubsection{Wage}

\begin{lemma}\label{derivative_wage_implicit}
  Derivative $\partial \bm{w} / \partial \tau$ can be analytically obtained. 
\end{lemma}
\begin{proof}
  Using the wage equation \eqref{wage_equation}, we derive the Jacobian matrix of the wage $\bm{w}$ with respect to the exogenous variable $\tau$.
This equation can be expressed in the following matrix form:
\begin{align}
\mu \, (\bm{w}^{1 - \sigma}\,
D)^\dagger \,
( \bm{\lambda} \bm{Y} \bm{\gamma}^{-1} ) ^* \,
=
\bm{w}, 
\end{align}
where $\bm{\gamma} = D^\top (\bm{\lambda} \bm{w}^{1 - \sigma})^*$ with $D = (\tau_{ij}^{1 - \sigma})_{i,j}$.
Applying the implicit function theorem,
 we obtain 
\begin{align}\label{Jacobian_wage}
  \frac{\partial \bm{w}}{\partial \tau} = 
  -
  \left( 
    \frac{\partial \bm{W}}{\partial \bm{w} }
  \right)^{-1}
    \frac{\partial \bm{W}}{\partial \tau },
\end{align}
where 
\begin{align}\label{Wexpression in w}
 & \bm{W} = \mu \, (\bm{w}^{1 - \sigma}\,
D)^\dagger \,
(\bm{\lambda} \bm{Y} \bm{\gamma}^{-1})^* \,
-
\bm{w}.
\end{align}
The wage $\bm{w}$ in equilibrium is the solution to the equation $\bm{W} = 0$.
The derivatives of $\bm{W}$ with respect to $\tau$
 and $\bm{w}$ are given as follows:
\begin{align}
  \frac{\partial \bm{W}}{\partial \tau } = &\,
 \mu \, \{ 
  [(\bm{w}^{1 - \sigma}\,
      D_{\tau})^\dagger \,
  (\bm{\lambda} \bm{Y} \bm{\gamma}^{-1})^*]^* -
  [(\bm{w}^{1 - \sigma}\, D)^\dagger
   (\bm{\lambda} \bm{Y} \bm{\gamma}^{-2})^* 
   D_{\tau}^\top (\bm{\lambda} \bm{w}^{1 - \sigma})^*
  ]^* \},
  \\
  \frac{\partial \bm{W}}{\partial \bm{w} } = &\, 
  \mu (1 - \sigma) \,
\mathrm{diag}\left[ (\bm{w}^{- \sigma} D)^\dagger \, (\bm{\lambda} \bm{Y} \bm{\gamma}^{-1})^*\right] 
\nonumber
\\ & +
\mu \, (\bm{w}^{1 - \sigma} \,
D)^\dagger \,
\left[
\left(\bm{\lambda} \bm{\gamma}^{-1} \frac{\partial \bm{Y}}{\partial \bm{w}} \right)^\dagger
- \left(\bm{\lambda} \bm{Y} \bm{\gamma}^{-2} \frac{\partial \bm{\gamma}}{\partial \bm{w}}
\right)^\dagger \right] - I,
\end{align}
where $D_{\tau} = \partial D / \partial \tau = ( \partial D_{ij} / \partial \tau )_{i,j}$. 
$\partial \bm{Y} / \partial \bm{w}$ in the above equation can be obtained by virtue of \eqref{income_wage}:
\begin{align}
      \frac{\partial \bm{Y}}{\partial \bm{w}} &= 
    \left[
      I - (1 - \mu)\,\bm{1}\bm{\lambda}^{\top}
    \right]^{-1},
    \label{Jacobian_income_with_wage} 
\end{align}
Moreover, we have 
\begin{align}
\frac{\partial \bm{\gamma}}{\partial \bm{w}} &= 
(1 - \sigma) D^\top 
\mathrm{diag}[(\bm{\lambda} \bm{w}^{- \sigma})^*].
\end{align}
\end{proof}

\subsubsection{Indirect Utility} 

We can express the indirect utility $\bm{v} = (v_i)_i$ in \eqref{qijhivi} as
\begin{align}
  \bm{v} = 
    \left(\frac{Y_i}{P_i^\mu r_i^{1 - \mu}}\right)_i
=   (\bm{P}^{-\mu} \bm{r}^{-(1 - \mu)} \bm{Y})^*,
\end{align}
where $\bm{P} = (P_i)_i$, $\bm{r} = (r_i)_i$, and $\bm{Y} = (Y_i)_i$.
The Jacobian matrix of $\bm{v}$ with respect to $\bm{\tau}$ is
\begin{align}\label{par_equ_indi}
\frac{\partial \bm{v}}{\partial \bm{\tau}}
&=
\left(\bm{P}^{-\mu} \bm{r}^{1-\mu} \frac{\partial \bm{Y}}{\partial \bm{\tau}}\right)^\dagger
- \mu \left(\bm{P}^{-(1+\mu)} \bm{r}^{\,1-\mu} \bm{Y} \frac{\partial \bm{P}}{\partial \bm{\tau}}\right)^\dagger 
- (1-\mu) \left(\bm{P}^{-\mu} \bm{r}^{-\mu} \bm{Y} \frac{\partial \, \bm{r}}{\partial \, \bm{\tau}}\right)^\dagger.
\end{align}
The three terms on the right hand side express the effects on indirect utility
of changes in income,  price index, and housing prices, 
respectively.
For $\tau = \tau_k$ $(k = 1, \ldots, p)$, the column vectors of 
the Jacobian matrices in the above equation are
given as follows:

\begin{lemma}\label{lemma_partial_derivative_Helpman}
  We have 
\begin{align}
  \frac{\partial \bm{Y}}{\partial \tau} & =
\left[
I - (1 - \mu)\,\bm{1}\bm{\lambda}^{\top}
\right]^{-1} 
 \frac{\partial \bm{w}}{\partial \tau},
\label{Jacobian_income_with_tau} \\
\frac{\partial \bm{r}}{\partial \tau} & = 
\frac{1 - \mu}{S} \, \left(\bm{\lambda} 
\frac{\partial \bm{Y}}{\partial \tau}\right)^*,
  \label{differentiation_rent}
\\
\frac{\partial \bm{P}}{\partial \tau} & =
  \frac{\kappa^{1 - \sigma}}{1 - \sigma}
\bm{P}^{\,\sigma} 
\left[
D_{\tau}^{\top} (\bm{\lambda}\bm{w}^{\,1-\sigma})^* + 
(1 - \sigma) \, D^\top 
\left(\bm{\lambda} \bm{w}^{-\sigma} 
\frac{\partial \bm{w}}{\partial \tau}\right)^*\right],
\label{P-differentiation-tau}
\end{align}  
where $D = (\tau_{ij}^{1 - \sigma})_{i,j}$, 
$D_{\tau} = \partial D / \partial \tau = ( \partial D_{ij} / \partial \tau )_{i,j}$, 
$\partial \bm{w} / \partial \tau$ is given in \eqref{Jacobian_wage}, 
and $\kappa$ is a constant:
\begin{align*}
\kappa
=
\frac{\sigma c}{\sigma - 1}
\left(
\frac{1}{f \sigma}\right)^{1 / (1-\sigma)}.
\end{align*}
\end{lemma}

\begin{proof}
By virtue of  \eqref{income_wage}, we have \eqref{Jacobian_income_with_tau}.
Moreover, using  \eqref{housing_price_equilibrium} yields the Jacobian matrix of the housing price $\bm{r}$ shown in  \eqref{differentiation_rent}.
Using Eq.~\eqref{price_index_equilibrium}, we rewrite the price index in the following matrix form:
\begin{align}
\bm{P}
=
\kappa
\left[
D^{\top} (\bm{\lambda}\,\bm{w}^{\,1-\sigma})^*
\right]^{1 / (1-\sigma)} \quad \Longrightarrow \quad
\bm{P}^{1-\sigma}
=
\kappa^{1-\sigma}
D^{\top} (\bm{\lambda}\,\bm{w}^{\,1-\sigma})^*
\end{align}
Using the above equation yields
\eqref{P-differentiation-tau}.
\end{proof}

\subsubsection{Trade Value} 

We decompose the change in the trade value $V_{ij}$ induced by a change in transportation cost $\tau_k$ into the effects of changes in price, the mass of varieties, demand, and population:
\begin{align}
  \frac{\mathrm{d} V_{ij}}{\mathrm{d} \tau_k} = 
  m_i q_{ij} \lambda_j \frac{\mathrm{d} p_{ij}}{\mathrm{d} \tau_k} + 
  p_{ij} q_{ij} \lambda_j \frac{\mathrm{d} m_i}{\mathrm{d} \tau_k} + 
  m_i p_{ij} \lambda_j \frac{\mathrm{d} q_{ij}}{\mathrm{d} \tau_k} + 
  m_i p_{ij} q_{ij} \frac{\mathrm{d} \lambda_j}{\mathrm{d} \tau_k}.
\end{align}
The derivatives of the above endogenous variables are given by 
\begin{align*}
  \frac{\mathrm{d} \Phi_{ij}}{\mathrm{d} \tau_k} &=  
  \frac{\partial \Phi_{ij}}{\partial \tau_k} + 
  \sum_{n \in N}
  \frac{\partial \Phi_{ij}}{\partial \lambda_n} \frac{\partial \lambda_n}{\partial \tau_k} \quad (\Phi=p,q),
  \\
  \frac{\mathrm{d} m_i}{\mathrm{d} \tau_k} &=  
  \frac{\partial m_i}{\partial \tau_k} + 
  \sum_{n \in N}
  \frac{\partial m_i}{\partial \lambda_n} \frac{\partial \lambda_n}{\partial \tau_k},
  \\
  \frac{\mathrm{d} \lambda_{j}}{\mathrm{d} \tau_k} &=   
  \frac{\partial \lambda_{j}}{\partial \tau_k},
\end{align*}
where 
\begin{align*}
  \frac{\partial p_{ij}}{\partial \tau_k} =& \,
  \frac{c \sigma}{\sigma - 1} \, \left( 
    w_i \frac{\partial \tau_{ij}}{\partial \tau_k} + 
    \tau_{ij} \frac{\partial w_i}{\partial \tau_k}
    \right),
  \\
  \frac{\partial p_{ij}}{\partial \lambda_n} =& \,
  \frac{c \sigma}{\sigma - 1} \, \left( 
    \tau_{ij} \frac{\partial w_i}{\partial \lambda_n}
    \right),
  \\
  \frac{\partial m_i}{\partial \tau_k} =&\, 0,
  \\
  \frac{\partial m_i}{\partial \lambda_n} =& 
  \begin{cases}
    1 / (\sigma f)  \quad (i = n), 
    \\ 
    0 \qquad \quad \ \,  (i \neq n), 
  \end{cases} 
  \\
  \frac{\partial q_{ij}}{\partial \tau_k} =&\, \mu 
  P_j^{-(1 - \sigma)}
  \left(
    p_{ij}^{-\sigma} \frac{\partial Y_j}{\partial \tau_k} - 
    \sigma Y_j p_{ij}^{-\sigma - 1} \frac{\partial p_{ij}}{\partial \tau_k}
  \right) 
  \nonumber
  \\ & 
  - 
  \mu (1 - \sigma) Y_j
  p_{ij}^{-\sigma}
  P_j^{-2(1 - \sigma)}
  \left(
    \sum_{o\in N}  m_o p_{oj}^{- \sigma} \frac{\partial p_{oj}}{\partial \tau_k}
  \right),
  \\
  \frac{\partial q_{ij}}{\partial \lambda_n} =&
  \mu 
  P_j^{-(1 - \sigma)}
  \left(
    p_{ij}^{-\sigma} \frac{\partial Y_j}{\partial \lambda_n} - 
    \sigma Y_j p_{ij}^{-\sigma - 1} \frac{\partial p_{ij}}{\partial \lambda_n}
  \right) 
  \nonumber
  \\ & 
  - 
  \mu Y_j
 p_{ij}^{-\sigma}
  P_j^{-2(1 - \sigma)}
  \left(
    \frac{p_{nj}^{1 - \sigma}}{\sigma f} + 
    (1 - \sigma) \sum_{o\in N}
    \frac{m_o}{p_{oj}^{\sigma}}
      \frac{\partial p_{oj}}{\partial \lambda_n}
  \right)
  .
\end{align*}
$\partial \bm{Y} / \partial \bm{\tau} = (\partial Y_i / \partial \tau_k)_{i,k}$ and 
$\partial \bm{w} / \partial \bm{\tau} = (\partial w_i / \partial \tau_k)_{i,k}$ 
are shown in Eqs.~\eqref{Jacobian_income_with_tau} and \eqref{Jacobian_wage}, respectively.
$\partial \bm{Y} / \partial \bm{\lambda} = (\partial Y_i / \partial \lambda_j)_{i,j}$ and $\partial \bm{w} / \partial \bm{\lambda} = (\partial w_i / \partial \lambda_j)_{i,j}$ in the above equations can be analytically obtained, as shown in the following lemma.
\begin{lemma}\label{}
  The derivatives $\partial \bm{Y} / \partial \bm{\lambda}$ and $\partial \bm{w} / \partial \bm{\lambda}$ can be analytically obtained. 
\end{lemma}
\begin{proof}
The derivatives $\partial \bm{Y} / \partial \bm{\lambda} = (\partial Y_i / \partial \lambda_j)_{i,j}$ 
and $\partial \bm{w} / \partial \bm{\lambda} = (\partial w_i / \partial \lambda_j)_{i,j}$ 
are derived by following the same procedure as in the proof of Lemma~\ref{derivative_wage_implicit}.
\begin{align}
  \frac{\partial \bm{Y}}{\partial \bm{\lambda}} &= -
  \left( 
    \frac{\partial \bm{H}}{\partial \bm{Y} }
  \right)^{-1}
    \frac{\partial \bm{H}}{\partial \bm{\lambda} }
,
  \label{Jacobi_income_population}
  \\
  \frac{\partial \bm{w}}{\partial \bm{\lambda}} &= -
  \left( 
    \frac{\partial \bm{W}}{\partial \bm{w} }
  \right)^{-1}
    \frac{\partial \bm{W}}{\partial \bm{\lambda} }
,
  \label{Jacobi_wage_population}
\end{align}
where 
$  \bm{H} = 
  \bm{w} - 
  \left[
        I - (1 - \mu) \bm{1} \, \bm{\lambda}^\top 
  \right] \bm{Y}$ 
  (cf. \eqref{income_wage}),
  $\bm{W}$ is given in 
  \eqref{Wexpression in w},
and
\begin{align*} 
  \frac{\partial \bm{H}}{\partial \bm{\lambda} } &=
  (1 - \mu) \bm{1} \bm{Y}^\top,
  \\
  \frac{\partial \bm{W}}{\partial \bm{\lambda} } &=
  \mu \, (\bm{w}^{1 - \sigma} \, D)^\dagger \,
\frac{\partial}{\partial  \bm{\lambda}}
(\bm{\lambda} \bm{Y} \bm{\gamma}^{-1})^* \\
 & = \mu \, (\bm{w}^{1 - \sigma} \,
D)^\dagger \,
\left\{
  \mathrm{diag}[ (\bm{Y} \bm{\gamma}^{-1})^*]  +
  \left(\bm{\lambda} \bm{\gamma}^{-1} \frac{\partial \bm{Y}}{\partial \bm{\lambda}}\right)^\dagger - 
  \left(\bm{\lambda} \bm{Y} \bm{\gamma}^{-2} \frac{\partial \bm{\gamma}}{\partial \bm{\lambda}}\right)^\dagger \,
\right\}
\end{align*}
with
$\frac{\partial \bm{\gamma}}{\partial \bm{\lambda}}
 = D^\top 
 \frac{\partial}{\partial \bm{\lambda}}
 [(\bm{\lambda}\bm{w}^{1 - \sigma})^*]=
D^\top  \mathrm{diag}(\bm{w}^{1 - \sigma})$.
\end{proof}


\section{Theoretical Details}\label{AppendixTheoryDetail}

We present theoretical details.
We often set $n^{\rm L}=m$ and suppress the superscript $(\cdot)^{\rm L}$
for simplicity.

\subsection{Transformation Matrices and Reduced Governing Equation}%
\label{TransMatH}

The coordinate transformation
\eqref{coordinate transformation}
can be rewritten 
using $\bm{a}=\bm{\lambda}^{\rm L}$
and 
$\bm{A}=\bm{F}^{\rm L}$
 as follows:
\begin{align} \label{coordinate transformationApp} &
{\bm \lambda}=H\begin{pmatrix}
{\bm a} \cr
{\bm b}
\end{pmatrix}=
(H_a,H_b)\begin{pmatrix}
{\bm a} \cr
{\bm b}
\end{pmatrix}
=H_a{\bm a}+H_b{\bm b},
\qquad 
\begin{pmatrix}
{\bm A} \cr
{\bm B}
\end{pmatrix}=\tilde{H}^\top \bm{F}=
\begin{pmatrix}
\tilde{H}_a^\top \bm{F} \cr
H_b^\top \bm{F}
\end{pmatrix} .
\end{align}
The submatrices $H_a$ and $\tilde{H}_a$ 
are defined by
\begin{align*}
 & H_a=
 \begin{scriptsize}
\begin{pmatrix}
\frac{1}{n_1}\bm{1}_{n_1} && \cr
  & \ddots & \cr
  && 
  \frac{1}{n_m}  \bm{1}_{n_m}
 \end{pmatrix}
 \end{scriptsize},
 \qquad
 \tilde{H}_a=
 \begin{scriptsize}
\begin{pmatrix}
\bm{1}_{n_1} && \cr
  & \ddots & \cr
  && 
 \bm{1}_{n_m}
 \end{pmatrix} ,
  \end{scriptsize}
 \end{align*}
where $\bm{1}_{n_\alpha}
  =(1,\ldots,1)^\top
  \in \mathbb{R}^{n_\alpha}$
$(\alpha=1,\ldots,m)$
  and $n_\alpha$ denotes
   the number of local places for the component $a_\alpha$
  of $\bm{a}$.
  The column vectors of $H_b$ are constructed
so as to be orthogonal to the columns of 
both $H_a$ and $\tilde{H}_a$,
   as follows: 
\begin{align*}
 & H_b=
\begin{pmatrix}
 W_1 && \cr
  & \ddots & \cr
  && 
 W_m
 \end{pmatrix}
\quad \mbox{with} \quad
 W_\alpha \equiv
\left(
\begin{scriptsize}
\begin{array}{cccccccccc}
n_\alpha - 1 & \cr
-1 & n_\alpha - 2\cr
\vdots & \ddots & \ddots \cr
\vdots & \ddots & -1 & 1 \cr
-1 & \cdots & \cdots & -1
\end{array}
\end{scriptsize}
\right).
\end{align*}

Then, we can transform the governing equation
\eqref{IncNonlinear} as
\begin{align}
 &\begin{pmatrix}
{\rm d}{\bm A} \cr {\rm d}{\bm B}
\end{pmatrix}= 
\tilde{H}^\top{\rm d}\bm{F}=
\tilde{H}^\top JH\begin{pmatrix}
{\rm d} {\bm a} \cr
{\rm d} {\bm b}
\end{pmatrix}
 +\tilde{H}^\top G{\rm d} \bm{\tau} +{\rm h.o.t.}=\bm{0}.
 \label{TransformedGovern}
\end{align}
By applying the relations 
\[
\tilde{H}^\top JH=
\begin{pmatrix}
J_{a} & J_{ab} \cr
J_{ba} & J_{b} 
\end{pmatrix},
\quad
\tilde{H}^\top G=
\begin{pmatrix}
G_a \cr
G_b
\end{pmatrix},
\]
we decompose \eqref{TransformedGovern} into
the following two equations:
\begin{align}
\label{Increment 1st} &
{\rm d}\bm{A}=
 J_{a} \, {\rm d}{\bm a}+  J_{ab} \, {\rm d}{\bm b}
    +  G_a {\rm d}\bm{\tau} +{\rm h.o.t.}=\bm{0}, 
    \\
 &{\rm d}\bm{B}=
  J_{ba} \, {\rm d} {\bm a} +  J_{b} \, {\rm d}{\bm b}
    +  G_b {\rm d}\bm{\tau}
    +{\rm h.o.t.}=\bm{0} \label{Increment 2nd}.
\end{align}
Under the condition that 
$J_{b}$ is nonsingular,
\eqref{Increment 2nd} can be solved to obtain 
\begin{align} \label{dbSolvedApp} &
{\rm d}{\bm b} = 
-J_{b}^{-1} J_{ba}{\rm d}{\bm a}  
 - J_{b}^{-1}  G_b{\rm d}\bm{\tau} +{\rm h.o.t.}
 \end{align}
Substituting this equation 
into \eqref{Increment 1st} yields \eqref{Reduced eq general}
in Lemma~\ref{country-levelCondense}
with ${J}^{\rm L}=J_{a}-J_{ab}J_{b}^{-1} J_{ba}$,
${G}^{\rm L}=G_a - J_{ab} J_{b}^{-1} G_b$,
$\bm{a}=\bm{\lambda}^{\rm L}$,
and 
$\bm{A}=\bm{F}^{\rm L}$.

  \subsection{Proof of Lemma~\ref{Lemma reductions}}%
  \label{Proof of Lemma reductions}

The relation $\sum_{\alpha=1}^{n^{\rm L}} \lambda_\alpha^{\rm L} =1$
in \eqref{L conservation}
gives the first condition 
$\sum_{\alpha=1}^{n^{\rm L}} {\rm d} \lambda_\alpha^{\rm L} =0$.
The relation $\sum_{i\in N} F_i =0$
 in \eqref{law for dynamics conservation law}
 leads to $\sum_{\alpha=1}^{n^{\rm L}}F_\alpha^{\rm L} =0$,
which further entails the second condition
$\sum_{\alpha=1}^{n^{\rm L}}{\rm d}F_\alpha^{\rm L} =0$
and the third condition
$\sum_{\alpha=1}^{n^{\rm L}}
 \frac{\partial F_\alpha^{\rm L}}{\partial \tau_k}=0.$

\subsection{Proof of \eqref{1D simplex}}\label{Proof 1D simplex}
\begin{align*} 
{\rm d}F_1-{\rm d}F_2  
 & = P^\top 
 \left\{
\begin{pmatrix}
J_{11} & J_{12} \cr
J_{21} & J_{22}
\end{pmatrix}
P {\rm d}\lambda_1
 +
Pg \,
{\rm d}\tau
+{\rm h.o.t.} \right\}
 \cr
    & = (1,-1)
\begin{pmatrix}  J_{11} & J_{12} \cr J_{21} & J_{22} \end{pmatrix}   
\begin{pmatrix} 1 \cr -1 \end{pmatrix} {\rm d}\lambda_1
  +(1,-1)
\begin{pmatrix}
1 \cr -1 \end{pmatrix} g \, {\rm d}\tau
+{\rm h.o.t.}
 \cr
    & =(J_{11} - J_{12} - J_{21} + J_{22}) \, {\rm d}\lambda_1
    +2g \, {\rm d}\tau
    +{\rm h.o.t.}=0.
\end{align*}

\subsection{Proof of Proposition~\ref{Prop for incremental}}%
\label{Proof of Prop for incremental}

In the target-level governing equation
 in \eqref{Reduced eq general},
we introduce the notation:
\begin{align*}
 & {\rm d} \tilde{\bm{\lambda}}=({\rm d}\lambda_1, \ldots, {\rm d}\lambda_{m-1})^\top, \cr
 & {\rm d}\tilde{\bm{F}} =
 \left( {\rm d}F_1 -{\rm d}F_m, \ldots, 
 {\rm d}F_{m-1}-{\rm d}F_m
 \right)^\top, \cr
& \bar{G}=\left\{ g_{\alpha k}
 \mid \alpha=1,\ldots,m-1;~k=1,\ldots,p \right\}
  \end{align*}
  ($m=n^{\rm L}$ and suppress $(\cdot)^{\rm L}$)
  and an $m\times (m-1)$ matrix:
\begin{align} \label{PmatrixDefine} &
P=
\begin{scriptsize}
\begin{pmatrix}
 1 \cr 
   & \ddots \cr
   &           & 1 \cr
    -1 & \cdots & -1
\end{pmatrix}.
\end{scriptsize}
\end{align}

The reduction of
${\bm{\lambda}}$, ${\bm{F}}$, and ${G}$
is given in Lemma~\ref{lemma for these reductions}. 

\begin{lemma}\label{lemma for these reductions}
${\rm d}\bm{\lambda}=P{\rm d}\tilde{\bm{\lambda}}$,  
${\rm d}\tilde{\bm{F}} =P^\top{\rm d}\bm{F}$, and
${G}=P\bar{G}$.
\end{lemma}
\begin{proof}
First,
from the first relation of \eqref{conservation m},
we have ${\rm d}\lambda_m
=- \sum_{\alpha=1}^{m-1} 
 {\rm d}\lambda_\alpha$.
 Then, 
 using the expression of $P$ in \eqref{PmatrixDefine},
 we obtain
\begin{align*}
{\rm d}\bm{\lambda} & =
\left({\rm d}\lambda_1, \ldots,{\rm d} \lambda_{m-1},
-{\rm d}\lambda_1
-\cdots-{\rm d}\lambda_{m-1}
\right)^\top
 = P({\rm d}\lambda_1, \ldots, {\rm d}\lambda_{m-1})^\top 
 =P{\rm d}\tilde{\bm{\lambda}}. 
\end{align*}
Next, using the second relation of \eqref{conservation m},
we can eliminate one equation and
introduce an $(m-1)$-dimensional vector:
\begin{align*} 
{\rm d}\tilde{\bm{F}} & =
\left( {\rm d}F_1 -
 {\rm d}F_m, \ldots, 
 {\rm d}F_{m-1}-
 {\rm d}F_m
 \right)^\top =P^\top{\rm d}\bm{F}.
\end{align*}
Finally, Lemma~\ref{Lemma reductions} and 
 $g_{\alpha k}= \frac{\partial F_\alpha}{\partial \tau_k}$ lead to
$ g_{m k}=- 
\sum_{\alpha=1}^{m-1} 
g_{\alpha k}$.
Then, we have
  \begin{align*}
P\bar{G} & 
=P\left\{ g_{\alpha k}
 \mid \alpha=1,\ldots,m-1;~k=1,\ldots,p \right\} \cr
 &
 = \begin{pmatrix} & \bar{G} & \cr
  -\sum_{\alpha=1}^{m-1} g_{\alpha 1} & \cdots &
    -\sum_{\alpha=1}^{m-1} g_{\alpha p}
 \end{pmatrix}
 = \begin{pmatrix} & \bar{G} & \cr
  g_{m1} & \cdots & g_{mp}
 \end{pmatrix} \cr
 & =\left\{ g_{\alpha k}
 \mid \alpha=1,\ldots,m;~k=1,\ldots,p \right\}={G}.
 \end{align*}
Hence, ${G}=P\bar{G}$.
\end{proof}

  Using ${\rm d}\bm{F}=
 {J}{\rm d}{\bm \lambda}
    + {G}{\rm d}\bm{\tau}+
    {\rm h.o.t.}$ given in \eqref{Reduced eq general}
    and 
    Lemma~\ref{lemma for these reductions},
we obtain
\begin{align*}
 {\rm d}\tilde{\bm{F}}  = P^\top {\rm d}\bm{F}  
  & =P^\top( {J}{\rm d}{\bm \lambda}
    + {G}{\rm d}\bm{\tau}
    +{\rm h.o.t.}) 
   =P^\top( {J} P{\rm d}\tilde{\bm \lambda}
    + P\bar{G}{\rm d}{\bm \tau}
    +{\rm h.o.t.}) \cr
     & =P^\top{J} P {\rm d}\tilde{\bm \lambda}
         + P^\top P \bar{G}{\rm d}{\bm \tau}
    +{\rm h.o.t.} =\mbox{$\bm{0}$}. 
\end{align*}
Setting 
$\tilde{J}=P^\top{J} P$ and
$\tilde{G}=P^\top P \bar{G}$
yields the relation \eqref{Gtau=Jd}
($J=J^{\rm L},~ P\bar{G}=G^{\rm L}$).

\OMIT{
\subsection{Proof of Proposition~\ref{PropEtReduced}}\label{ProofEtReduced}

Under the condition that 
$J_{b}$ is nonsingular,
\eqref{dbSolvedApp} is substituted into
\eqref{coordinate transformationApp}
to arrive at
\begin{align*} 
{\rm d}{\bm \lambda} &
=H_a \, {\rm d}{\bm a}+H_b \,{\rm d}{\bm b} \cr &
= H_a\, {\rm d}{\bm a}+H_b(-J_{b}^{-1} J_{ba}{\rm d}{\bm a}  
 - J_{b}^{-1}  G_b{\rm d}\bm{\tau} +{\rm h.o.t.})
 \cr &
 =Q{\rm d}{\bm a}+R{\rm d}\bm{\tau} +{\rm h.o.t.}
\end{align*}
with 
\[
Q=H_a-H_b J_{b}^{-1} J_{ba}, \quad
R=-H_b J_{b}^{-1} G_b.
\]
Then,
$\frac{\partial \bm{\lambda}}{\partial \bm{a}}=Q$.
}

\subsection{Trade Value for Each Country}\label{TradeValueEachCountry}

We define the trade value for each country.
Let $\mathcal{C}_\alpha$ denote the set of places in country $\alpha$ with $n_\alpha$ places.
The set of all places can be expressed as $N=\{1, \ldots, n \} = \bigcup_{\alpha = 1}^{n^C} \mathcal{C}_\alpha$, where $n^C$ is the number of the countries.
For 
$i \in \mathcal{C}_\beta$ 
and $j \in \mathcal{C}_\alpha$, 
$T_{ij}$ denotes the 
value of exports from place $i$ in country $\beta$ to place $j$ in country $\alpha$.
Using this definition, the trade values
 from country $\beta$ to $\alpha$ is given by
\begin{align} \label{Tvolume-define}
  V_{\beta\rightarrow\alpha} = 
  \sum_{i \in \mathcal{C}_\beta}
  \sum_{j \in \mathcal{C}_\alpha}
   V_{ij}.
\end{align}


\end{document}